\documentclass[conference,letterpaper]{ieeeconf} 

\IEEEoverridecommandlockouts                              

\usepackage{amsmath,amssymb}
\usepackage{graphicx}
\usepackage{latexsym,url,subfigure}
\usepackage{comment}
\usepackage{stfloats}

\newtheorem{theorem}{Theorem}

\newtheorem{lemma}{Lemma}

\newtheorem{definition}{Definition}
\newtheorem{corollary}{Corollary}

\newcommand{\grad}{{\rm grad}}
\newcommand{\N}{\mathcal{N}}

\newcommand{\V}{{\mathcal V}}
\newcommand{\E}{{\mathcal E}}
\newcommand{\Q}{{\mathcal Q}}
\newcommand{\M}{{\mathcal M}}

\newcommand{\F}{{\mathcal F}}
\newcommand{\R}{{\mathbb R}}

\newcommand{\tr}{{\rm tr}}
\newcommand{\sk}{{\rm sk}}

\newcommand{\diag}{{\rm diag}}

\newcommand{\ewi}{R_{wi}}

\renewcommand{\L}{{\cal L}}

\newcommand{\TSO}{T_{\ewi}SO(3)}

\title{3-D Visual Coverage Based on
Gradient Descent Techniques on Matrix Manifold
and Its Application to Moving Objects Monitoring}

\author{Takeshi Hatanaka, Riku Funada and Masayuki Fujita
\thanks{Takeshi Hatanaka, Riku Funada and Masayuki Fujita are with the Department of Mechanical and 
Control Engineering, Tokyo Institute of Technology, Tokyo 152-8552, JAPAN
        {\tt\small {hatanaka,funada,fujita}@ctrl.titech.ac.jp}}}

\begin{document}
\maketitle
\thispagestyle{empty}
\pagestyle{empty}

\begin{abstract}
This paper investigates coverage control 
for visual sensor networks based on gradient descent
techniques on matrix manifolds.
We consider the scenario 
that networked vision sensors with controllable orientations
are distributed over 3-D space to 
monitor 2-D environment.
Then, the decision variable must be constrained
on the Lie group $SO(3)$.
The contribution of this paper is two folds.
The first one is technical, namely we formulate the coverage problem as an optimization
problem on $SO(3)$ without introducing local parameterization like
Eular angles and directly apply the gradient descent algorithm on 
the manifold.
The second technological contribution is to 
present not only the coverage control scheme but
also the density estimation process including image processing
and curve fitting while exemplifying its effectiveness 
through simulation of moving objects monitoring.
\end{abstract}

\section{Introduction}

Aspirations for safety and security of human lives against 
crimes and natural disasters 
motivate us to establish smart monitoring systems 
to monitor surrounding environment.
In this regard, vision sensors are expected as powerful 
sensing components since they
provide rich information about the outer world. 
Indeed, visual monitoring systems have been already
commoditized and are working in practice.
Typically, in the current systems, 
various decision-making and situation awareness processes 
are conducted at a monitoring center by human operator(s),
and partial distributed computing at each sensor is, if at all, done 
independently of the other sensors.
However, as the image stream increases, 
it is desired to distribute the entire process
to each sensor while achieving total optimization
through cooperation among sensors.

Distributed processing over the visual sensor networks 
is actively studied in recent years motivated by a variety of
application scenarios 
\cite{amit2}--\cite{HWF_CDC13}.
Among them, several papers address optimal monitoring of the environment
assuming mobility of the vision sensors
\cite{EYE}--\cite{HWF_CDC13},
where it is required for the network to ensure
the best view of a changing environment \cite{EYE}.
The problem is related to coverage control \cite{CL_EJC05}--\cite{BCM_ES05},
whose objective is to deploy mobile sensors efficiently
in a distributed fashion.
A typical approach to coverage control is to employ
the gradient descent algorithm
for an appropriately designed aggregate objective function.
The objective function is usually formulated by
integrating the product 
of a sensing performance function of a point
and a density function indicating the 
relative importance of the point.
The approach is also applied to visual coverage 
in \cite{EYE}--\cite{GTF_CDC08}.
The state of the art of coverage control is compactly summarized in \cite{EYE},
and a survey of related works in the computer vision society is found in \cite{survey}.

In this paper, we consider a visual coverage problem 
under the situation where vision sensors with controllable 
orientations are distributed over the 3-D space
to monitor 2-D environment.
In the case, the control variables i.e. the rotation matrices
must be constrained on the Lie group $SO(3)$,
which distinguishes the present paper from the works
on 2-D coverage \cite{cortes}--\cite{ZM_SIAM13}.
On the other hand, \cite{EYE,DSMFC_SP12,HWF_CDC13} consider
situations similar to this paper.
\cite{DSMFC_SP12,HWF_CDC13} take game theoretic approaches
which allow the network to achieve globally optimal coverage with
high probability but instead the convergence speed tends
to be slower than the standard gradient descent approach.
In contrast, \cite{EYE} employs the gradient approach
by introducing a local parameterization of
the rotation matrix and regarding the problem as optimization on a vector space.

This paper approaches the problem differently from \cite{EYE}.
We directly formulate the problem as optimization on $SO(3)$
and apply the gradient descent algorithm on matrix manifolds \cite{AMS_BK}.
This approach will be shown to allow one to parametrize the 
control law for a variety of underactuations imposed by the hardware constraints.
This paper also addresses density estimation from acquired data,
which is investigated in \cite{IJRR} for 2-D coverage.
However, we need to take account of the following 
characteristics of vision sensors:
(i) the sensing process inherently includes
projection of 3-D world onto 2-D image, and
(ii) explicit physical data is not provided.
To reflect (i), we incorporate the projection into the 
optimization problem on the embedding manifold of $SO(3)$.
The issue (ii) is addressed technologically, where 
we present the entire process including image processing 
and curve fitting techniques.
Finally, we demonstrate the utility of the present coverage 
control strategy through simulation of moving objects monitoring.

\subsection*{Preliminary: Gradient on Riemannian Manifold}

Let us consider a Riemannian manifold $\M$ whose
tangent space at $x \in \M$ is denoted by $T_x\M$,
and the corresponding Riemannian metric, an smooth inner product, 
defined over $T_x\M$ is denoted by $\langle \cdot,\cdot\rangle_x$.
Now, we introduce a smooth scalar field $f(\cdot): \M \to \R$ defined over the manifold $\M$,
and the derivative of $f$ at an element $x \in \M$ in the direction $\xi \in T_x\M$,
denoted by $Df(x)[\xi]$.
We see from Definition 3.5.1 and (3.15) of \cite{AMS_BK}
that the derivative $Df(R)[\Xi]$ is defined by
\[
Df(x)[\xi] = \left.\frac{d f(\gamma(t))}{dt}\right|_{t = 0},
\]
where $\gamma: \R \to \M$ is a smooth curve
such that $\gamma(0) = x$.
In particular, when $\M$ is a linear manifold with $T_x\M = \M$,
 the derivative $Df(R)[\Xi]$ is equal to the classical directional derivative
\begin{equation}
Df(x)[\xi] = \lim_{t\to 0}\frac{f(x + t\xi)-f(x)}{t}.
\label{eqn:derivative_lin_M}
\end{equation}

Now, the gradient of $f$
is defined as follows.
\begin{definition}\cite{AMS_BK}
\label{def:grad_M}
Given a smooth scalar field $f$ defined over a Riemannian manifold $\M$,
the gradient  of $f$ at $x$, denoted by $\grad_x^{\M} f$, is defined as the unique
element of $T_x\M$ satisfying
\[
\langle \grad_x^{\M} f, \xi\rangle_x = Df(x)[\xi]\ \ {\forall \xi}\in T_x\M.
\]
\end{definition}
Suppose now that $\M$ is a Riemannian 
submanifold of a Riemannian manifold $\N$,
namely $T_x\M$ is a subspace of $T_x \N$ and
they share a common Riemannian metric.
In addition, the orthogonal projection of an element of $T_x \N$
onto $T_x\M$ is denoted by $P_x: T_x \N \to T_x \M$.
Then, the following remarkable lemma holds true.
\begin{lemma}\cite{AMS_BK}
\label{lem:grad_proj}
Let $\bar f$ be a scalar field defined over $\N$
such that the function $f$ defined on $\M$ 
is a restriction of $\bar f$.
Then, the gradient of $f$ satisfies the equation
\begin{equation}
\grad_x^{\M}f = P_x \grad_x^{\N} \bar f.
\label{eqn:grad_proj}
\end{equation}
\end{lemma}



\begin{figure}[t]
\begin{center}
\includegraphics[width=7cm]{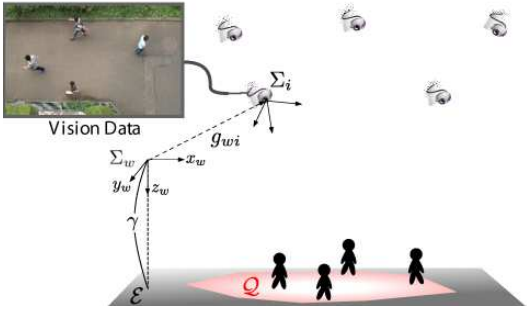}
\caption{Targeted scenario.}
\label{fig:scenario}
\end{center}
\end{figure}

\section{Targeted Scenario}


\subsection{Vision Sensors and Environment}

We consider the situation illustrated in Fig. \ref{fig:scenario} where
$n$ vision sensors $\V = \{1,\cdots, n\}$ are located in 3-D Euclidean space.
Let the fixed world frame be denoted by 
$\Sigma_w$ and the body fixed frame of sensor $i \in \V$
by $\Sigma_i$.
We also denote the position vector of the origin of $\Sigma_i$
relative to $\Sigma_w$ by $p_{wi} \in \R^3$,
and the rotation matrix of $\Sigma_i$ relative to $\Sigma_w$
by $\ewi\in SO(3) := \{R \in \R^{3\times 3}|\ RR^T = R^TR = I_3,\
\det(R) = +1\}$.
Then, the pair $g_{wi} = (p_{wi},\ewi)\in SE(3) := \R^3 \times SO(3)$,
called {\it pose}, represents
the configuration of sensor $i$.
In this paper, each sensor's position $p_{wi}$ is assumed to be fixed,
and sensors can control only their orientations $\ewi$.
In addition, we suppose that sensors are localized and calibrated {\it a priori}
and $g_{wi}$ is available for control.

We use the notation $g_{wi}$ to describe not only the pose
but also a coordinate transformation operator similarly to \cite{MLS_BK}.
Take two frames $\Sigma_a$ and $\Sigma_b$. 
Let the pose of the frame $\Sigma_b$ relative to $\Sigma_a$ be denoted by $g_{ab} = (p_{ab},R_{ab})$,
and the coordinates of a point relative to $\Sigma_a$ by $p_b$.
Then, the coordinates $p_a$ of the point relative to $\Sigma_a$ are 
given as
\[
p_a = g_{wi}(p_b) :=
R_{ab}p_b + p_{ab}.
\]

Let us next define the region to be monitored by a group of sensors $\V$.
In this paper, we assume that 
the region is a subset of a 2-D plane (Fig. \ref{fig:scenario}), where the 
2-D plane is called the environment and the subset to be monitored
is called the mission space.
Let the set of coordinates of all points in the environment and the mission space 
relative to $\Sigma_w$ are respectively 
denoted by $\E$ and $\Q$.
Just for simplicity, we suppose that 
the world frame $\Sigma_w$ is attached so that its
$x,y$-plane is parallel to the environment
(Fig. \ref{fig:scenario}).
Then, the set $\E$ is formulated as 
\[
\E = \{q \in \R^3|\ {\bf e}_3^Tq = \gamma\}
\]
with some constant $\gamma \in \R$,
where ${\bf e}_i \in \R^3,\ i=1,2,3$ is an $i$-th standard basis.
Suppose that a metric $\phi: \E \to \R_+$,
called a density function, indicating the relative importance of 
every point $q\in \E$ is defined over $\E$.
In this paper, the function $\phi(q)$ is assumed to be small if point $q$ is important
and to satisfy $\phi(q) = \bar{\phi} \ {\forall q} \notin \Q$ with a constant
$\bar{\phi}$ such that $\bar{\phi} > \sup_{q\in \Q}\phi(q)$.

\subsection{Geometry}

A vision sensor has an image plane containing the sensing array, 
whose elements, called pixels, 
provide the numbers reflecting 
the amount of light incident.
We assume that the image plane is a rectangle as illustrated in Fig. \ref{fig:image_plane1}.
The set of position vectors of all points on the image plane relative to the sensor frame $\Sigma_i$
is denoted by ${\mathcal F}_i \subset \R^3$. 
Now, the axes of the sensor frame $\Sigma_i$ is assumed to be selected 
so that its $x,y$-plane is parallel to the image plane
and $z$-axis perpendicular to the image plane passes through 
the focal center of the lens.
Then, the third element of any point in the set ${\mathcal F}_i$ 
must be equal to the focal length
$\lambda_i$.

\begin{figure}[t]
\begin{center}
\begin{minipage}{4cm}
\begin{center}
\includegraphics[width=4cm]{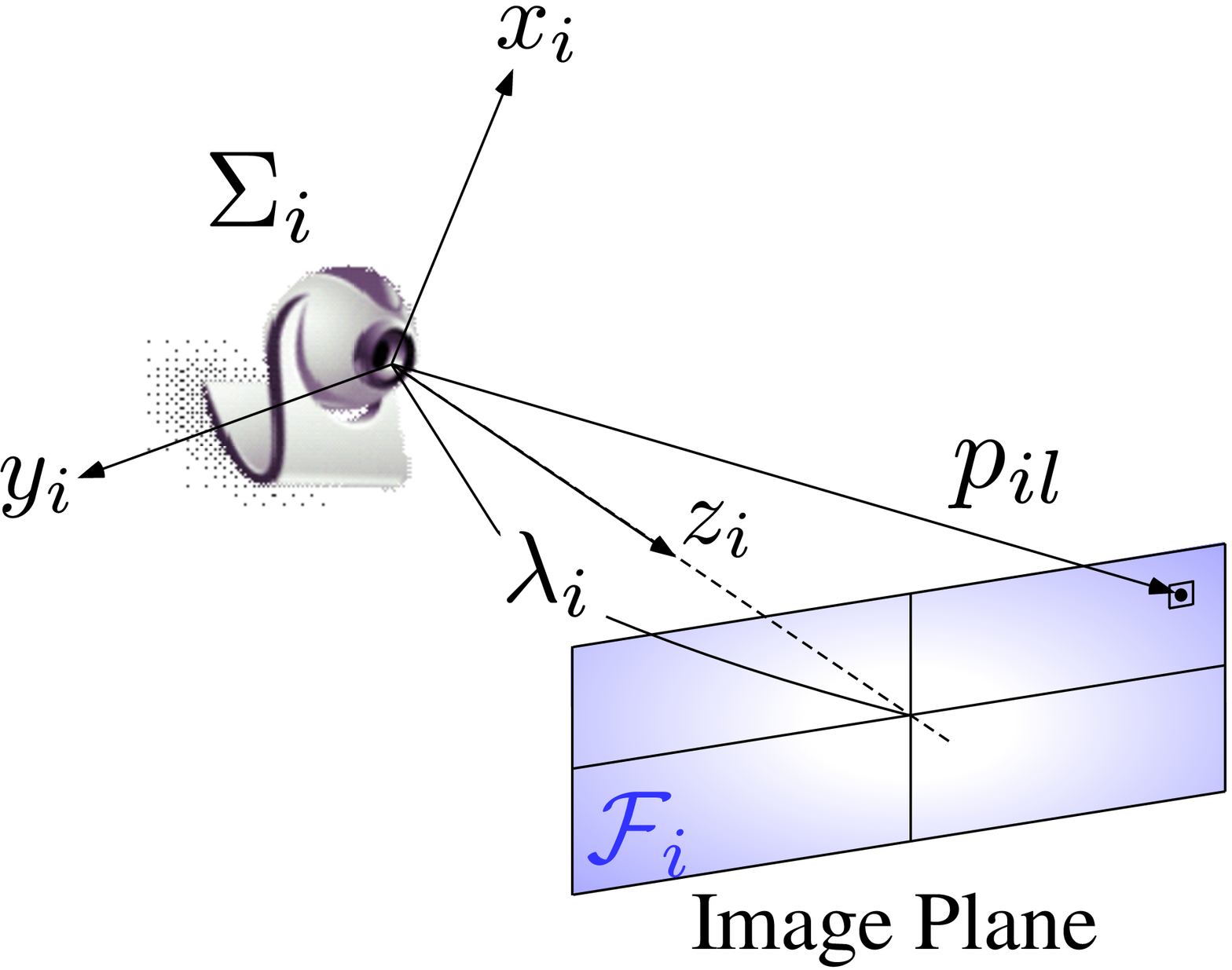}
\caption{Image plane and pixel.}
\label{fig:image_plane1}
\end{center}
\end{minipage}
\hspace{.2cm}
\begin{minipage}{4cm}
\begin{center}
\includegraphics[width=4cm]{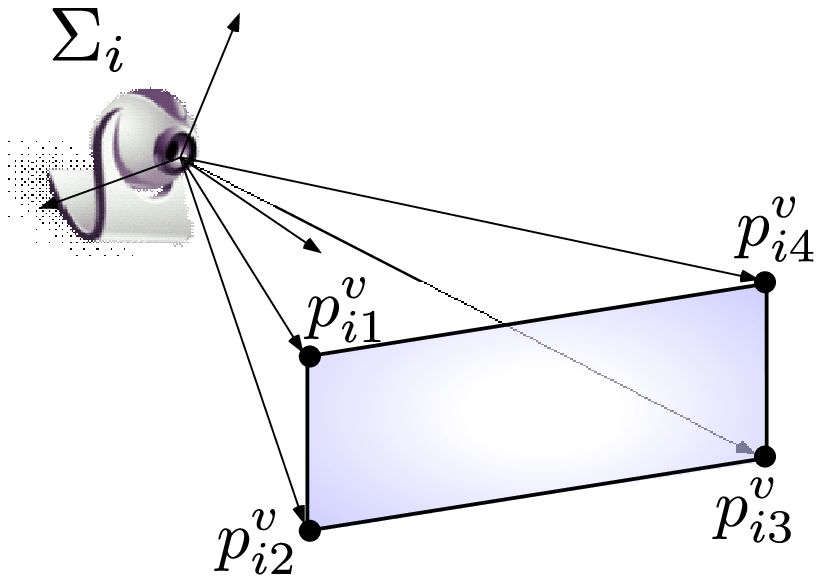}
\caption{Vertices of image plane.}
\label{fig:image_plane2}
\end{center}
\end{minipage}
\end{center}
\end{figure}

We next denote the set of pixels of sensor $i\in \V$ by $\L_i = \{1,\cdots, L_i\}$ and
the position vector of the center of the $l$-th pixel on the image plane of sensor $i$
relative to $\Sigma_i$ by $p_{il} \in {\mathcal F}_i$.
Since $l$ in $p_{il}$ and $p_{jl}$ deffer,
we may need to use the notation like $l_i$ but we omit the subscript
to reduce notational complexity.
In addition, the positions of its vertices relative to 
$\Sigma_i$ are denoted by $p^{v}_{i1}, \cdots, p^{v}_{i4} \in {\mathcal F}_i$
(Fig. \ref{fig:image_plane2}).

When a point on the environment with coordinates $q_i$ relative to $\Sigma_i$
is captured by a sensor $i$ with $g_{wi}$,
the point is projected onto the image plane as illustrated in
Fig. \ref{fig:proj1}.
If the coordinates of the projected point 
are denoted by $q_i^{\rm im} \in {\mathcal F}_i$,
it is well known that the projection is formulated as
\begin{equation}
q_i^{\rm im} = \Psi_i(q_i) = 
\frac{\lambda_i}{{\bf e}_3^Tq_i}
q_i.
\end{equation}
It is not difficult to show that the inverse map $\Phi_i$ 
of the map $\Psi_i$ (Fig. \ref{fig:proj1}) from $q^{\rm im}_i$
to $q_i$ is given by
\begin{equation}
q_i = \Phi_i(q^{\rm im}_{i}) =
\frac{\delta_iq^{\rm im}_{i}}{{\bf e}_3^T \ewi q^{\rm im}_{i}}.
\label{eqn:phi}
\end{equation}
Note that, while $\Psi_i$ is independent of $\ewi$,
the map $\Phi_i$ depends on $\ewi$ and hence we describe $\Phi_i$ as 
$\Phi_i(q^{\rm im}_{i}; \ewi)$.

\begin{figure}[t]
\begin{center}
\includegraphics[width=6.5cm]{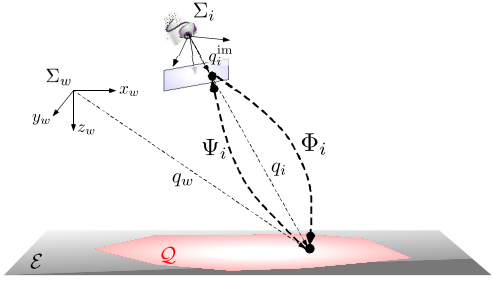}
\caption{Projections $\Psi_i$ and $\Phi_i$.}
\label{fig:proj1}
\end{center}
\end{figure}

Using the map $\Phi_i$, we denote by $FOV_i(\ewi)$ the set of coordinates
of the field of view (FOV) of each sensor $i$ relative to $\Sigma_w$, which is also a polytope.
Its $\V$-polytope representation is trivial, namely it is given by
the convex hull of the four points with coordinates
\begin{equation}
{q}^{v}_{wl}(\ewi) = g_{wi}\circ \Phi_i(p^{v}_{il}; \ewi)
\label{eqn:arxiv1}
\end{equation}
relative to $\Sigma_w$
(Fig. \ref{fig:FOV2}).
The ${\mathcal H}$-polytope representation is also 
computed efficiently as follows.

Suppose now that $l$-th side line segment ($l = 1,2,3,4$) specifying 
the boundary of the image plane connects the vertices $1$ and $2$
without loss of generality.
Then, the line projected onto the environment is also a 
line segment whose vertices have coordinates 
${p}^{v}_{w1}(\ewi)$ and ${p}^{v}_{w2}(\ewi)$
relative to $\Sigma_w$, and hence the line is formulated as
\[
\left\{q \in \E\left|\ A^i_l(\ewi) 
\begin{bmatrix}
q_1\\
q_2
\end{bmatrix}
 = 1,\ q_3 = \gamma \right. \right\},
\]
where the matrix $A_l(\ewi) \in \R^{2 \times 2}$ is derived as
\[
A^i_l(\ewi) = \begin{bmatrix}
1& 1
\end{bmatrix}
\begin{bmatrix}
\begin{bmatrix}
{\bf e}_1^T\\
{\bf e}_2^T
\end{bmatrix}
{q}^{v}_{w1}(\ewi) &\begin{bmatrix}
{\bf e}_1^T\\
{\bf e}_2^T
\end{bmatrix} {q}^{v}_{w2}(\ewi)
\end{bmatrix}^{-1}
\]
from the fact that ${q}^{v}_{w1}(\ewi)$ and ${q}^{v}_{w2}(\ewi)$
are on the line.
Since the coordinates $p_w^0 = g_{wi}\circ \Phi_i(p_i^0; \ewi)$
for any interior $p_i^0$ of $\F_i$ must be inside the FOV,
a half space specifying the FOV is described by the inequality
$\bar A_l^i(\ewi) q \leq \bar a_l^i(\ewi)$ with
\[
\bar A_l^i(\ewi) := \bar a_l^i(\ewi)A_l^i(\ewi),\ 
\bar a_l^i := {\rm sign} (1 - A_l^i(\ewi)p_w^0).
\]
In the same way, we can find the pair $\bar A_l^i(\ewi),\ 
\bar a_l^i$ for all $l = 1,2,3,4$.
Stacking them allows one to formulate the FOV as
\[
FOV_i(\ewi) = 
\left\{q \in \E\left|\ A^i(\ewi) 
\begin{bmatrix}
q_1\\
q_2
\end{bmatrix}
 = a^i(\ewi),\ q_3 = \gamma \right. \right\}.
\]


\section{Coverage for a Single Sensor}

In this section, we consider a simple case with 
$\V = \{i\}$.

\subsection{Objective Function}

Let us first define the objective function
to be {\it minimized} by sensor $i$.
In this paper, we basically take the concept of coverage control \cite{CL_EJC05,BCM_BK09}, 
where the objective function is defined by 
a sensing performance function and a density function at 
a point $q \in \E$.
Note however that
we accumulate the function only at the center of the pixels
projected onto the environment $\E$ in order to reflect
the discretized nature of the vision sensors.
In the sequel, the sensing performance function and the density function
at $q \in \E$ are denoted by
$f_i(q): \E \to \R_+$ and 
$\phi(q): \E \to \R_+$, respectively.

Let us first define a function $q_{wl}(\ewi): SO(3) \to \E$ 
providing the coordinates in $\Sigma_w$ of the point on $\E$
which is captured by $l$-th pixel as
\begin{eqnarray}
q_{wl}(\ewi) \!\!&\!\!=\!\!&\!\! g_{wi}\circ \Phi_i(p_{il};\ewi)  
= \frac{\delta_i \ewi p_{il}}{{\bf e}_3^T \ewi p_{il}} + p_{wi}.
\label{eqn:q_wl}
\end{eqnarray}
Then, the objective function takes the form of
\begin{eqnarray}
H_i(\ewi) \!\!&\!\!=\!\!&\!\! \sum_{l \in \L_i} w_{il} (f_i\circ q_{wl}(\ewi)) (\phi \circ q_{wl}(\ewi)),
\label{eqn:obj_single}
\end{eqnarray}
where $w_{il} > 0$ is a weighting coefficient.
If we impose a large $w_{il}$ on the pixel at around the center of the image,
the sensor tends to capture the important area at around the image center.
If we need to accelerate computation, 
replacing $\L_i$ in (\ref{eqn:obj_single}) by its subset
is an option.
In order to ensure preciseness, we need to introduce an extended function
allowing $\pm \infty$, 
but we will not mention it
since it can be easily avoided by choosing
$f_i(q)$ appropriately.

\begin{figure}[t]
%
%
\begin{center}
\includegraphics[width=6.5cm]{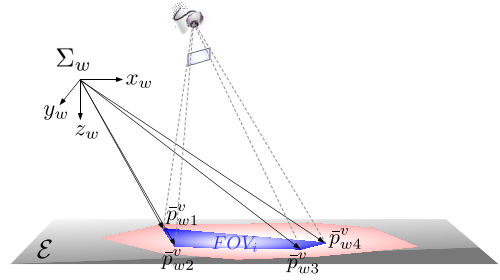}
\caption{Field of view $FOV_i$}
\label{fig:FOV2}
\end{center}
\end{figure}

Similarly to \cite{BCM_BK09}, we let
the performance function $f_i(q)$ depend only on the distance $\|q - p_{wi}\|$.
Remark however that, differently from \cite{BCM_BK09}, 
the third element of $q - p_{wi}$ is not controllable
since the sensor is fixed.
This may cause a problem that penalty of seeing distant area
does not work in the case that the element is large enough.
However, the element is not ignorable since it reflects heterogeneous 
characteristics of vision sensors in the multi-sensor case.
We thus use the weighting distance as
\begin{equation}
f_i(q) = \frac{1}{\lambda_i} \|q - p_{wi}\|_{W}^2 = \frac{1}{\lambda_i}(q-p_{wi})^TW(q - p_{wi}).
\label{eqn:perf1}
\end{equation}
with $W \geq 0$, where $1/\lambda_i$ is introduced 
since the distance is scaled by the focal length. 
Suppose that $W$ is set as $W = \diag([w\ w\ 1])$.
Then, a large $w$ imposes a heavy penalty on viewing distant area
and a small $w$ a light penalty on it.
In particular, when $q = q_{wl}(\ewi)$ for some $l \in \L_i$,
(\ref{eqn:perf1}) is rewritten as
\begin{equation}
f_i \circ q_{wl}(\ewi) 
= \frac{ \tilde \delta_i \|\ewi p_{il}\|_W^2}{\|{\bf e}_3^T \ewi p_{il}\|^2},\ \
\tilde \delta_i = \frac{\delta_i^2}{\lambda_i}.
\label{eqn:perf2}
\end{equation}

Once the density function $\phi$ is given,
the goal is reduced to minimization of
(\ref{eqn:obj_single}) with (\ref{eqn:perf2})
under the restriction of $\ewi \in SO(3)$.
In order to solve the problem,
this paper takes the gradient descent approach
which is a standard approach to coverage control.
For this purpose,
it is convenient to define an extension $\bar H_i: \R^{3 \times 3} \to \R_+$
such that $\bar H_i(M) = H_i$ if $M \in SO(3)$.
We first extend the domain
of $q_{wl}(\cdot)$ in (\ref{eqn:q_wl}) from $SO(3)$
to $\R^{3 \times 3}$ as
\begin{equation}
\bar q_{wl}(M) = 
\frac{\delta_i M p_{il}}{{\bf e}_3^T M p_{il}} + p_{wi}.
\label{eqn:arxiv2}
\end{equation}
Then, the vector $\bar q_{wl}(M) \in \R^3$ is not always 
on the environment when $M \notin SO(3)$ but
the function $f_i$ in (\ref{eqn:perf1}) is well-defined
even if the domain is altered from $\E$ to $\R^3$. 
We thus denote the function with the domain $\R^3$ by $\bar f_i$,
and define the composite function 
\begin{equation}
\bar{f}_i \circ \bar q_{wl}(M) = 
\frac{ \tilde \delta_i \|Mp_{il}\|_W^2}{\|{\bf e}_3^T M p_{il}\|^2}.
\label{eqn:perf5}
\end{equation}

\begin{figure}[t]
\begin{center}
%
%
\includegraphics[width=6cm]{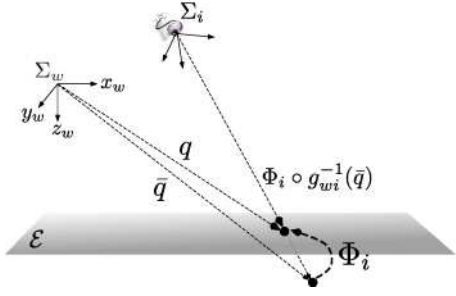}
\caption{Coordinates of $q \in \E$ relative to $\Sigma_w$ and $\Sigma_i$.}
\label{fig:proj4}
\end{center}
\end{figure}

We next focus on the term $\phi \circ q_{wl}(\ewi)$ in
(\ref{eqn:obj_single}) and expand the domain of the composite function
from $SO(3)$ to $\R^{3\times 3}$.
Here, since $\bar q_{wl}(M)$ is not always on $\E$,
we need to design $\bar \phi: \R^3 \to \R_+$ such that
$\bar \phi(q) = \phi(q)$ if $q \in \E$.
In this paper, we assign to a point $\bar q \in \R^3$ 
the density of a point 
\[
q = g_{wi}\circ \Phi \circ g_{wi}^{-1} (\bar q) 
= \frac{\delta_i(\bar q - p_{wi})}{{\bf e}_3^T(\bar q - p)} + p_{wi},
\]
where the operations are illustrated in Fig. \ref{fig:proj4}.
Accordingly, the density function is defined by
\begin{eqnarray}
\bar \phi(\bar q) = \phi\circ g_{wi}\circ \Phi \circ g_{wi}^{-1} (\bar q).
\label{eqn:perf6}
\end{eqnarray}
Remark that, differently from $f_i$, the function $\phi$
is not naturally extended and the selection of $\bar \phi$ is not unique. 
The motivation to choose (\ref{eqn:perf6}) will be clear in the 
next subsection.

Consequently, we define the extended objective function
\begin{eqnarray}
\bar H_i(M) \!\!&\!\!=\!\!&\!\! \sum_{l \in \L_i} w_{il}(\bar f_i\circ \bar q_{wl}(M)) (\bar \phi \circ \bar q_{wl}(M)),
\label{eqn:obj_single_fict}
\end{eqnarray}
from $\R^{3\times 3}$ to $\R_+$ by using
(\ref{eqn:perf5}) and (\ref{eqn:perf6}).
Let us finally emphasize that $H_i(M) = \bar H_i(M)$ holds for any $M\in SO(3)$.

\subsection{Density Estimation for Moving Objects Monitoring}

In the gradient descent approach,
we update the rotation $\ewi$
in the direction of $\grad_{\ewi[k]}^{SO(3)} H_i$
at each time $k$.
This subsection assumes that
the density $\phi$ is not given {\it a priori}
and that $\phi$ needs to be estimated from acquired vision data
as investigated in \cite{EYE,IJRR}.

Let us first consider an ideal situation such that the density function
is exactly projected onto the image plane, namely
\begin{equation}
\phi(q) = \psi \circ \Psi_i \circ g_{wi}^{-1}[k](q)\ {\forall q} \in FOV_i(\ewi[k]),
\label{eqn:psi_phi2}
\end{equation}
holds with respect to the density $\psi: \F_i \to \R_+$ over the image plane.
Then, the density function value $\phi(q)$ is available at
any point in the FOV.
We next consider a point $\bar q \in \R^3$
which does not always lie on $\E$. 
Then, the value of $\bar \phi(\bar q)$ is also given by 
the same function as (\ref{eqn:psi_phi2}) since
\begin{eqnarray}
\bar \phi(\bar q) \!\!&\!\!=\!\!&\!\! \phi\circ g_{wi}[k]\circ \Phi \circ g_{wi}^{-1}[k] (\bar q)
\nonumber\\
\!\!&\!\!=\!\!&\!\! \psi \circ \Psi_i \circ g_{wi}^{-1}[k] \circ g_{wi}[k]\circ \Phi \circ g_{wi}^{-1}[k] (\bar q)
\nonumber\\
\!\!&\!\!=\!\!&\!\! \psi \circ \Psi_i \circ \Phi \circ g_{wi}^{-1}[k] (\bar q)
= \psi \circ \Psi_i \circ g_{wi}^{-1}[k] (\bar q).
\nonumber
\end{eqnarray}
Ensuring the equality is the reason for choosing 
(\ref{eqn:perf6}).

\begin{figure}
\begin{center}
\includegraphics[width=7.5cm]{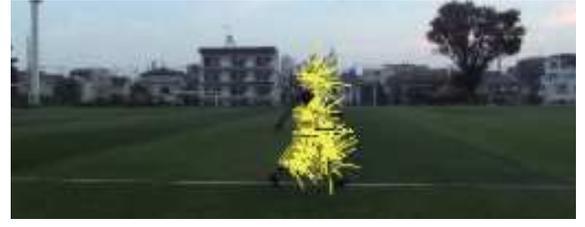}
\caption{A snapshot and computed optical flows.}
\label{fig:snapshot}
\end{center}
\end{figure}



We next consider estimation of the density $\psi$ on the image
since assuming (\ref{eqn:psi_phi2}) is unrealistic. 
Rich literature has been devoted to the
information extraction from the raw vision data,
and a variety of algorithms are currently available 
even without expert knowledge \cite{MATLAB1}.
For example, it is possible to detect and localize in the image plane
specific objects like cars or human faces, 
and even abstract targets such as everything moving 
or some environmental changes.

The present coverage scheme is indeed applicable to any scenario
such that a nonnegative number $y_{il}$ reflecting 
its own importance is assigned to each pixel $l \in \L_i$
after conducting some image processing.
However, we mainly focus on a specific scenario of
monitoring moving objects on the mission space.
Suppose that a sensor captures a human walking from 
left to right in the image as in Fig. \ref{fig:snapshot}.
Then, a way to localize such moving objects is to
compute optical flows from consecutive images as in Fig. \ref{fig:snapshot},
where the flows are depicted by yellow lines.
We also let the data $y_{il}$ be the norm of the flow vector
at each pixel.
Then, the plots of $y_{il}$ over the image plane 
are illustrated by green dots in  Fig. \ref{fig:mix_gauss1}.

\begin{figure}
\begin{center}
\begin{minipage}{4.2cm}
{\includegraphics[width=4.2cm]{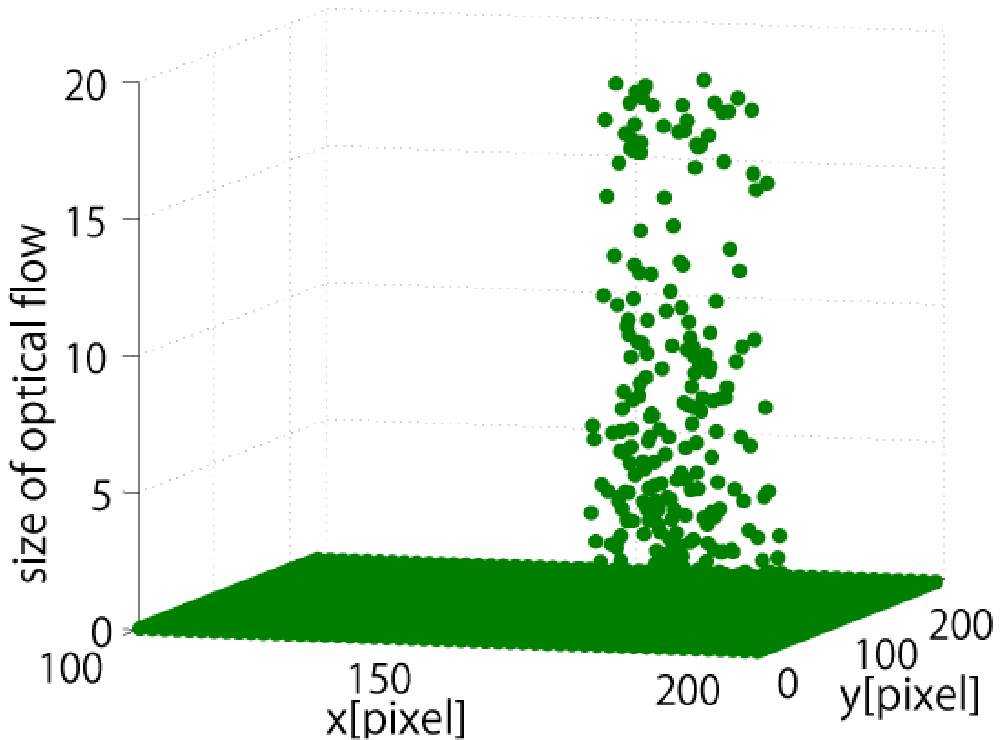}}
\caption{Plots of $y_{il}$.}
\label{fig:mix_gauss1}
\end{minipage}
\begin{minipage}{4.2cm}
{\includegraphics[width=4.2cm]{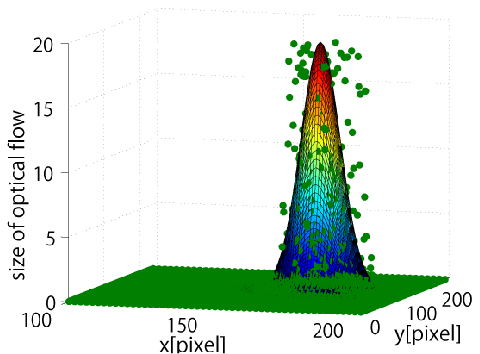}}
\caption{Estimated density.}
\label{fig:mix_gauss2}
\end{minipage}
\end{center}
\end{figure}


We next fit the data of $y_{il}$ by a continuous function
defined over $\F_i$ and use the function as $\psi$.
Such algorithms are also available
even in real time \cite{MATLAB2}.
Similarly to \cite{IJRR}, we employ the mixed Gaussian function
known to approximate a variety of functions with excellent precision
by increasing the number of Gaussian functions,
and widely used in data mining, pattern recognition,
machine learning and statistical analysis.
Fig. \ref{fig:mix_gauss2} shows the Gaussian function with $m=1$ computed so as to fit the data
in Fig. \ref{fig:mix_gauss1}.
Of course, using a larger $m$ achieves a better approximation
as shown in Fig. \ref{fig:mix_gauss_various}.

As a result, we obtain a function in the form of
\begin{equation}
\sum_{j=1}^m \alpha_j e^{-\|p^{\rm im} - \mu^{\rm im}_j\|^2_{\Sigma_j^{\rm im}}},\ \
\Sigma_j^{\rm im} > 0
\label{eqn:mix_Gauss_image2}
\end{equation} 
over the 2-D image plane coordinates $p^{\rm im} \in \R^2$.
Note that (\ref{eqn:mix_Gauss_image2}) is large when
$p^{\rm im}$ captures an important point, which is opposite to 
the density function.
Thus, we define the function
\begin{equation}
\psi^{\rm im}(p^{\rm im}) = \bar{\psi} -
\sum_{j=1}^m \alpha_j e^{-\|p^{\rm im} - \mu^{\rm im}_j\|^2_{\Sigma_j^{\rm im}}},
\label{eqn:mix_Gauss_image}
\end{equation} 
where $\bar{\psi} < \bar{\phi}$ is a positive scalar guaranteeing $\psi^{\rm im}(p^{\rm im}) \geq 0$
for all $p^{\rm im}$.
It is also convenient to
define $\psi(p)$ for all
3-D vectors $p\in \F_i$ on the image plane as
\begin{eqnarray}
\psi (p) \!\!&\!\!=\!\!&\!\!
\left\{
\begin{array}{l}
\bar{\phi},\ \mbox{if }g_{wi} \circ \Phi_i(p) \notin \Q \\
\bar{\psi} - \sum_{j=1}^m \alpha_j 
e^{-\|p - \mu_j\|^2_{\Sigma_j}},\ 
\mbox{otherwise}
\end{array}
\right.,
\label{eqn:mix_Gauss}\\
\mu_j \!\!&\!\!=\!\!&\!\! 
\begin{bmatrix}
\mu^{\rm im}_j\\
\lambda_i
\end{bmatrix},\ \
\Sigma_j =
\begin{bmatrix}
\Sigma_j^{\rm im}&0\\
0&0
\end{bmatrix}.
\end{eqnarray}

\begin{figure}
\begin{center}
\begin{minipage}{4cm}
{\includegraphics[width=4cm,height=2.8cm]{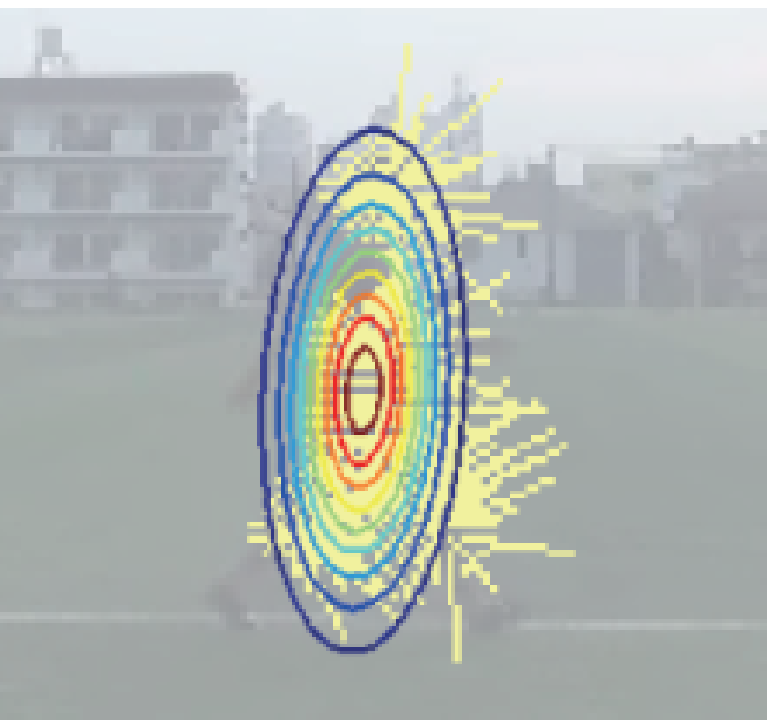}}
\end{minipage}
\hspace{0.2cm}
\begin{minipage}{4cm}
{\includegraphics[width=4cm,height=2.8cm]{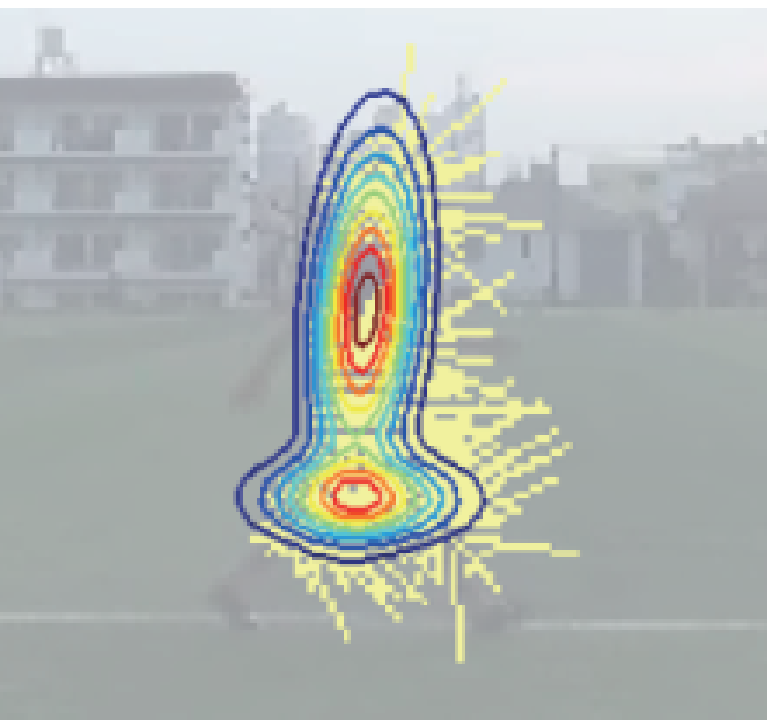}}
\end{minipage}
\caption{Estimated densities for $m = 1$(left) and $m = 2$(right).}
\label{fig:mix_gauss_various}
\end{center}
\end{figure}

\subsection{Gradient Computation}

\subsubsection*{Full 3-D Rotational Motion}

Here, we will derive the gradient 
$\grad_{\ewi[k]}^{SO(3)} H_i$, given a rotation $\ewi[k] \in SO(3)$ and $\psi$ in (\ref{eqn:mix_Gauss}).
It is widely known that $\TSO$ is formulated as
$\TSO = \{\ewi X \in \R^{3\times 3}|\ X \in so(3)\}$,
where $so(3)$ is the set of all skew symmetric matrices in $\R^{3\times 3}$.
We also define the operator $\wedge$(wedge) from $\R^3$ to $\R^{3\times 3}$
such that $a \times b = \hat a b$ for the cross product $\times$.
%
%
%
The rotational group $SO(3)$ is known to be a submanifold of
a Riemannian manifold $\R^{3 \times 3}$ with 
$T_x\R^{3\times 3} = \R^{3\times 3} \supset \TSO$
and the Riemannian metric 
\begin{eqnarray}
\langle M, N \rangle = \tr(M^TN),\ M, N\in \R^{3\times 3}
\label{eqn:Riemannian}
\end{eqnarray}
\cite{AMS_BK}.
It is also known that
the orthogonal projection $P_{\ewi}$ of matrix
$M \in T_{\ewi}\R^{3\times 3} = \R^{3\times 3}$ onto $\TSO$ 
in terms of the Riemannian metric induced by (\ref{eqn:Riemannian})
is given by
\begin{eqnarray}
P_{\ewi}(M) = \ewi\sk(\ewi^T M),\ \ \sk(M) = \frac{1}{2}(M - M^T).
\label{eqn:proj_so(3)}
\end{eqnarray}
See Subsection 3.6.1 of \cite{AMS_BK} for more details.

Now, we have the following theorem, where we use the notation
$\tilde \L_i(\ewi) = \{l \in \L_i|\ g_{wi} \circ \Phi(p_{il};\ewi) \in \Q\}$
and $\tilde \L_i^c(\ewi) = \L_i \setminus \tilde \L_i(\ewi)$.

\begin{theorem}
Suppose that the objective function $\bar H_i$ is formulated by (\ref{eqn:obj_single_fict})
with (\ref{eqn:perf5}), (\ref{eqn:perf6}) and (\ref{eqn:mix_Gauss}).
Then, the gradient $\grad_{\ewi[k]}^{SO(3)} H_i$ 
is given by 
\begin{eqnarray}
\grad_{\ewi[k]}^{SO(3)} H_i \!\!&\!\!=\!\!&\!\! 
P_{\ewi[k]}\left(\grad_{\ewi[k]}^{\R^{3\times 3}} \bar H_i\right),
\label{eqn:proj2}\\
\grad_{\ewi[k]}^{\R^{3\times 3}} \bar H_i \!\!&\!\!=\!\!&\!\!
\tilde \delta_i 
\eta_i^T(\ewi[k]) p_{il}^T,
 \nonumber
\end{eqnarray}
where
\begin{eqnarray}
\eta_i(R) \!\!&\!\!=\!\!&\!\!\!\!
\sum_{l \in \tilde \L_i^c(R)}w_{il}
\bar{\phi}  \eta_i^l+\!\!
\sum_{l \in \tilde \L_i(R)}\!\! w_{il}
\Big(
\bar{\psi} \eta_i^l
- \sum_{j=1}^m \alpha_j \eta_{i}^{lj}
\Big)
\nonumber\\
\tilde \L_i(R) \!\!&\!\!=\!\!&\!\! 
\{l \in \L_i|\ q_{wl}(R) \in \Q\},\ 
\tilde \L_i^c(R)  =  \L_i\setminus \tilde \L_i(R)
\nonumber\\
\eta_i^l(R) \!\!&\!\!=\!\!&\!\! \frac{2}{({\bf e}_3^T R p_{il})^3}
\Big(
({\bf e}_3^T R p_{il})p_{il}^TR^TW
- \|Rp_{il}\|_W^2 {\bf e}_3^T
\Big)
\nonumber\\
\eta_{i}^{lj}(R) \!\!&\!\!=\!\!&\!\!
\frac{2e^{-\|b_{lj}\|^2_{\Sigma_j}}}
{\lambda_i({\bf e}_3^TR p_{il})^3}
\Big(({\bf e}_3^TR p_{il})\xi^{lj}_i(R)
-\lambda_i\|Rp_{il}\|_W^2{\bf e}_3^T
\Big)
\nonumber\\
\xi^{lj}_i(R) \!\!&\!\!=\!\!&\!\!
 \|Rp_{il}\|_W^2b_{lj}^T \Sigma_j (p_{il} {\bf e}_3^T - \lambda_i I_3)R^T
+ \lambda_i p_{il}^TR^TW
\nonumber
\end{eqnarray}
\end{theorem}

\begin{proof}
See Appendix \ref{app:1}.
\end{proof}

Namely, just running the dynamics 
\begin{equation}
\dot{R}_{wi} = - K \grad_{\ewi}^{SO(3)}  H_i,\ K > 0
\label{eqn:grad_descent1}
\end{equation}
leads $\ewi$ to the set of critical points of $H_i$.
However, in practice, 
the vision data is usually obtained at discrete time instants
and hence 
we approximate the continuous-time algorithm (\ref{eqn:grad_descent1}) by
\begin{eqnarray}
{R}_{wi}[k+1] = \ewi[k]{\rm exp}\left(\ewi^T[k] \left(\alpha_k \grad_{\ewi[k]}^{SO(3)} H_i\right)\right).
\label{eqn:grad_descent2}
\end{eqnarray}
See \cite{AMS_BK} for the details on the selection of $\alpha_k$.

\subsubsection*{Rotational Motion with Underactuations}

\begin{figure}[t]
\begin{center}
\begin{minipage}{4cm}
\begin{center}
\includegraphics[width=4cm]{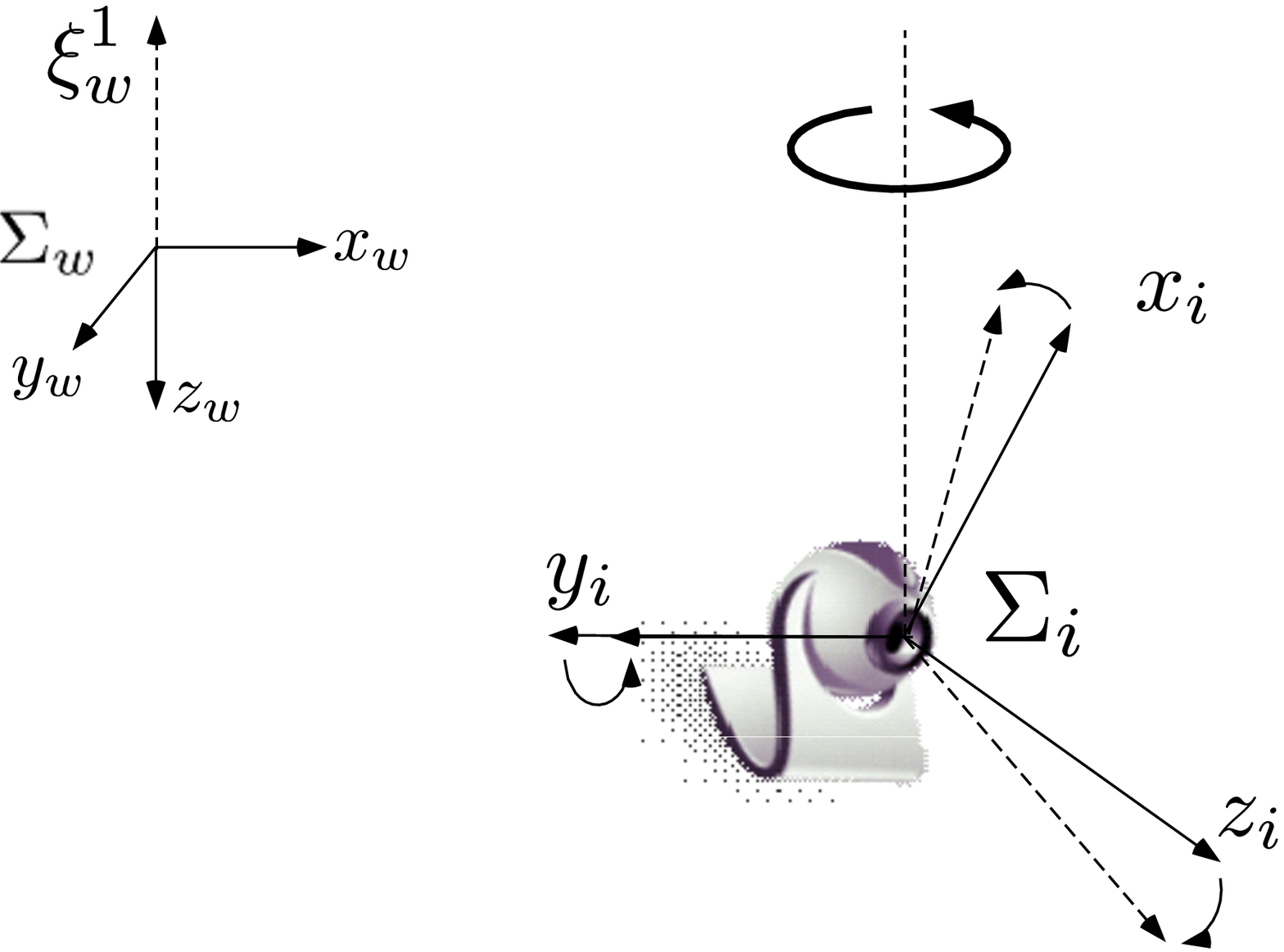}
\caption{Pan motion.}
\label{fig:Pan}
\end{center}
\end{minipage}
\hspace{.2cm}
\begin{minipage}{4cm}
\begin{center}
\includegraphics[width=4cm]{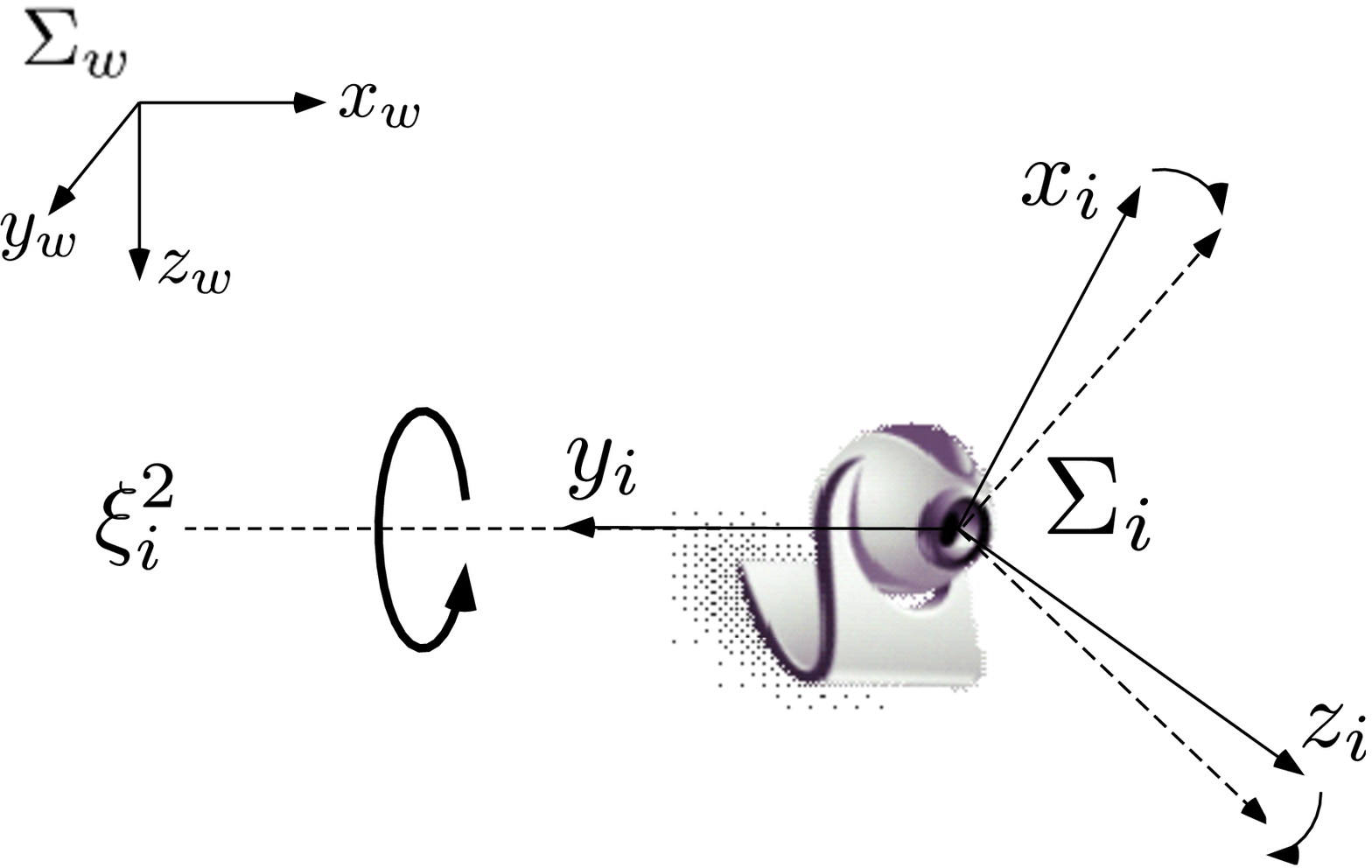}
\caption{Tilt motion.}
\label{fig:Tilt}
\end{center}
\end{minipage}
\end{center}
\end{figure}

In the above discussion, we assume that the sensor can take full 3-D
rotational motion.
However, the motion of many commoditized cameras is restricted by
the actuator configurations.
Hereafter, we suppose that the sensor can be rotated 
around two axes $\xi^1_i$ ($\|\xi^1_i\|=1$) and $\xi_i^2$ ($\|\xi^2_i\|=1$),
where these vectors are defined in $\Sigma_i$ and assumed to be
linearly independent of each other.
These axes may depend on the rotation matrix $\ewi$.
For example, in the case of
Pan-Tilt (PT) cameras in Figs. \ref{fig:Pan} and \ref{fig:Tilt}, which are
typical commoditized cameras, 
the axis of the pan motion (Fig. \ref{fig:Pan}) 
is fixed relative to $\Sigma_w$, 
while that of the tilt motion (Fig. \ref{fig:Tilt}) is fixed
relative to the sensor frame $\Sigma_i$.
Then, only one of the two axes depends on $\ewi$.
Note that even when there is only one axis around which the sensor can be rotated, 
the subsequent discussions are valid just letting $\xi_i^2=0$.


Let us denote a normalized vector $\xi^3_i$ ($\|\xi_1^3\|=1$)
orthogonal to the $\xi^1_i, \xi^2_i$-plane.
Then, 
the three vectors $\xi^1_i$, $\xi^2_i$ and 
$\xi_i^3$ span $\R^3$.
Thus, any element $\Theta$ of $\TSO$ can be represented in the form of
$\Theta = \ewi \sum_{j=1}^3 \beta_j \hat{\xi}_i^j,\ \ \beta_j \in \R$.
Now, we define a {\it distribution} $\Delta$ \cite{MLS_BK} assigning
$\ewi \in SO(3)$ to the subspace 
\begin{equation}
\left\{\Theta\in \TSO\left|\ \Theta = \ewi \sum_{j=1}^2 \beta_j \hat{\xi}_i^j,\ \beta_1, \beta_2 \in \R
\right.
\right\},
\label{eqn:submanifold}
\end{equation}
whose dimension is $2$.
The distribution $\Delta$ is clearly regular and hence 
induces a submanifold ${\mathcal S}_{UA}$ of $SO(3)$ \cite{MLS_BK},
called integral manifold, such that its tangent space 
$T_{\ewi}{\mathcal S}_{UA}$ at $\ewi\in {\mathcal S}_{UA}
\subseteq SO(3)$ is equal to (\ref{eqn:submanifold}).
The manifold ${\mathcal S}_{UA}$ specifies 
orientations which the camera can take.

Since ${\mathcal S}_{UA}$ is a submanifold of $SO(3)$,
a strategy similar to Theorem 1 is available and
we have the following corollary.

\begin{corollary}
Suppose that the objective function $\bar H_i$ is formulated by (\ref{eqn:obj_single_fict})
with (\ref{eqn:perf5}), (\ref{eqn:perf6}) and (\ref{eqn:mix_Gauss}).
Then, the gradient $\grad_{\ewi[k]}^{{\mathcal S}_{UA}} H_i$ 
is given by 
\begin{eqnarray}
\grad_{\ewi[k]}^{{\mathcal S}_{UA}} H_i \!\!&\!\!=\!\!&\!\! 
P^{UA}_{\ewi[k]}\left(\grad_{\ewi[k]}^{SO(3)} H_i\right)
\label{eqn:proj3}
\end{eqnarray}
where the orthogonal projection $P^{UA}_{\ewi}(M)$ of
$M = \ewi N\in \TSO,\ N \in so(3)$ to $T_{\ewi}{\mathcal S}_{UA}$ is defined by
\begin{eqnarray}
P^{UA}_{\ewi}(M) \!\!&\!\!=\!\!&\!\! \ewi (
\alpha_1 \hat{\xi}_i^1 
+ \alpha_2 \hat{\xi}_i^2)
\label{eqn:proj4}
\end{eqnarray}
with $\alpha_l =
\frac{\langle N, \hat{\xi}_i^l\rangle - 
\langle \hat{\xi}_i^1, \hat{\xi}_i^2\rangle
\langle N, \hat{\xi}_i^{-l}\rangle}
{1 - \langle \hat{\xi}_i^1, \hat{\xi}_i^2\rangle^2}$
$l = 1, 2$, where $-l = 2$ if $l=1$ and
$-l = 1$ if $l=2$.
\end{corollary}

We see from this corollary  
that the gradient
$\grad_{\ewi[k]}^{SO(3)} H_i$ on $SO(3)$ 
is utilized as it is and we need only to project it
through (\ref{eqn:proj4}).
Also, the projection (\ref{eqn:proj4}) is successfully 
parameterized by the vectors $\xi_i^1$ and $\xi_i^2$.


\section{Coverage for Multiple Sensors}

\begin{figure}[t]
\begin{center}
\includegraphics[width=7cm]{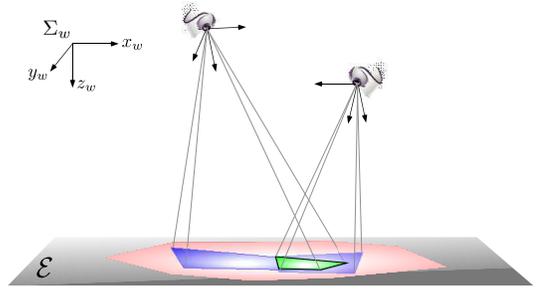}
\caption{Overlaps of FOVs (green region).}
\label{fig:overlap}
\end{center}
\end{figure}

In this section, we extend the result of the previous section to the
multi-sensor case.
The difference from the single sensor case stems from the overlaps
of the FOVs with the other sensors as illustrated in Fig. \ref{fig:overlap}.
\cite{EYE,CL_EJC05,BCM_ES05} present sensing performance functions
taking account of the overlaps and their gradient decent laws.
However, in this paper, we present another simpler scheme to manage the overlap.

Let us first define the set of sensors capturing a point $q \in \E$ within the FOV as
${\mathcal V}(q;R_{\V}) = \{i\in \V|\ q \in FOV_i(\ewi)\}$
where $R_{\V} = (\ewi)_{i\in \V}$.
We also suppose that, when $\V(q;R_{\V})$ has multiple elements for some $q \in Q$,
only the data of the sensor with the minimal 
sensing performance (\ref{eqn:perf1})
among sensors in $\V(q;R_{\V})$ is employed in higher-level decisions and recognitions.
This motivates us to partition $FOV_i(\ewi)$ into the two region 
\begin{eqnarray}
SFOV_i(R_{\V}) \!\!&\!\!=\!\!&\!\! \{q \in FOV_i(\ewi)|\ i \in \arg\min_{j \in \V(q;R_{\V})} f_j(q)\},
\nonumber
\label{eqn:SFOV}\\
SFOV^c_i(R_{\V}) \!\!&\!\!=\!\!&\!\! \{q \in FOV_i(\ewi)|\ i \notin \arg\min_{j \in \V(q;R_{\V})} f_j(q)\}.
\nonumber
\label{eqn:SFOVc}
\end{eqnarray}
Then, what pixel $l$ captures a point in $SFOV^c_i(R_{\V})$
is identified with what it captures a point outside of $\Q$,
whose cost is set as $\bar \phi (f_i\circ q_w)$ in the previous section,
in the sense that both of the data are not useful at all. 
This is reflected by assigning $\bar{\phi}$ to the pixels $l \notin \tilde \L_i(R_{\V})$
with
$\tilde \L_i(R_{\V})= \{l\in \L_i|\ q_{wl}(\ewi) \in SFOV_i(R_{\V}) \cap \Q\}$.

Accordingly, we formulate the  function
to be minimized by $\V$ as $H(R_{\V}) = 
\sum_{i\in \V} H_i(R_{\V})$ with
\begin{eqnarray}
H_i(R_{\V}) \!\!&\!\!=\!\!&\!\! 
\sum_{l \in \tilde \L_i(R_{\V})} w_{il} (f_i\circ q_{wl}(\ewi)) (\phi \circ q_{wl}(\ewi))
\nonumber\\
&&\hspace{1cm} + \bar{\phi} \sum_{l \notin \tilde \L_i(R_{\V})} w_{il} f_i\circ q_{wl}(\ewi).
\label{eqn:obj_multi}
\end{eqnarray}
Remark that (\ref{eqn:obj_multi}) differs from
(\ref{eqn:obj_single}) only in the set $\tilde \L_i(R_{\V})$.

Strictly speaking, to compute the gradient of (\ref{eqn:obj_multi}),
we need to expand $\L_i(R_{\V})$ from $SO(3)\times \cdots \times SO(3)$ to 
$\R^{3\times 3}\times \cdots \times \R^{3\times 3}$. 
For this purpose, it is sufficient to define 
$\bar{FOV}_i(M)$ from $\R^{3\times 3}$ to a subset of $\E$.
For example, an option is to define 
an extension
\begin{equation}
\bar q^v_{wl}(M) = 
\frac{\delta_i M p^v_{il}}{{\bf e}_3^T M p^v_{il}} + p_{wi},\ l=1,2,3,4.
\label{eqn:arxiv3}
\end{equation}
of (\ref{eqn:arxiv1}) similarly to (\ref{eqn:arxiv2}), and
to let $\bar{FOV}_i(M)$ be the convex full of these points.
However, at the time instants computing the gradient with $R_{\V}[k]$,
the extension $\bar \L_i(M_{\V})$ for a sufficiently small
perturbation $M_{\V}-R_{\V}$ is equivalent to
the original set $\tilde \L_i(R_{\V})$
irrespective of the selection of $\bar{FOV}_i(M)$
except for the pathological case when a pixel is located on the boundary
of $SFOV_i(R_{\V})$.
Namely, ignoring such pathological cases which do not happen
almost surely for (\ref{eqn:grad_descent2}),
the gradient can be computed by using 
the set $\tilde \L_i(R_{\V}[k])$ instead of its extension.
Hence, the gradient is simply given as Theorem 1
by just replacing $\tilde \L_i(\ewi)$ by $\tilde \L_i(R_{\V})$.
Note that the curve fitting process is run without 
taking account of whether $l\in \tilde \L_i(\ewi)$ or not, 
and $\bar{\phi}$ is assigned
to $l\notin \tilde \L_i(\ewi)$ 
at the formulation of $\psi$ as in (\ref{eqn:mix_Gauss}).
This is because letting
$y_{il} = 0\ {\forall l} \notin \tilde \L_i(\ewi)$
at the curve fitting stage would degrade the density estimation
accuracy at around the boundary of $SFOV_i$.


%

The remaining issue is efficient computation of the set $\tilde \L_i(\ewi)$.
Hereafter, we assume that each sensor acquires $FOV_j\ {\forall j}\in \V\setminus \{i\}$,
i.e. $\bar A_i^l(R_{wj})$ and $\bar a_i^l(R_{wj})$ for all $l=1,2,3,4$,
and its index $j$
through (all-to-all) communication or with the help of a centralized computer.
The computation under the limited communication will be
mentioned at the end of this section.
In addition, we suppose that every sensor stores the set
\begin{equation}
\Q_{ij} = \left\{q \in \Q\left|\ \lambda_j\left\|q - 
p_{wi}
\right\|_W^2
>
\lambda_i \left\|q - 
p_{wj}
\right\|_W^2
\right.\right\},
\label{eqn:3Dvoronoi}
\end{equation}
for all $j\in \V$ which can be computed off-line
since the sensor positions are fixed.

\begin{figure}[t]
\begin{center}
\begin{minipage}{4cm}
\begin{center}
\includegraphics[width=4cm,height=2.5cm]{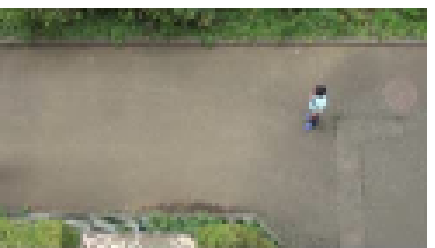}
\end{center}
\end{minipage}
\hspace{.2cm}
\begin{minipage}{4cm}
\begin{center}
\includegraphics[width=4cm,height=2.5cm]{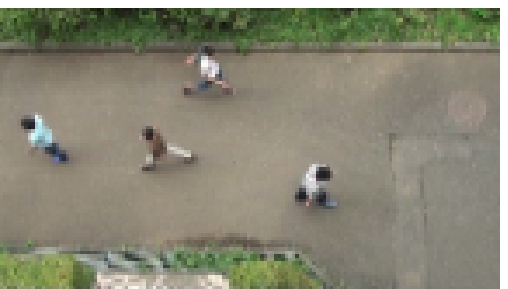}
\end{center}
\end{minipage}
\caption{Images at $t=0$s (left) and $t = 4$s (right).}
\label{fig:initial_image}
\bigskip

\begin{minipage}{4cm}
\begin{center}
\includegraphics[width=4cm,height=2.5cm]{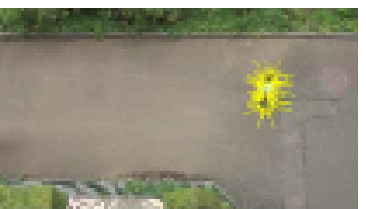}
\end{center}
\end{minipage}
\hspace{.2cm}
\begin{minipage}{4cm}
\begin{center}
\includegraphics[width=4cm,height=2.5cm]{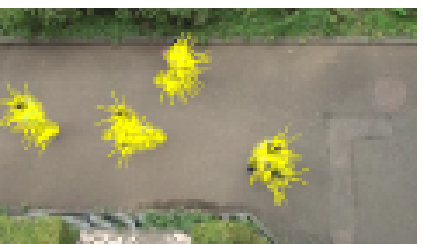}
\end{center}
\end{minipage}
\caption{Optical flows of images at $t=0$s (left) and $t = 4$s (right).}
\label{fig:of_image}
\bigskip

\begin{minipage}{4cm}
\begin{center}
\includegraphics[width=4cm,height=2.5cm]{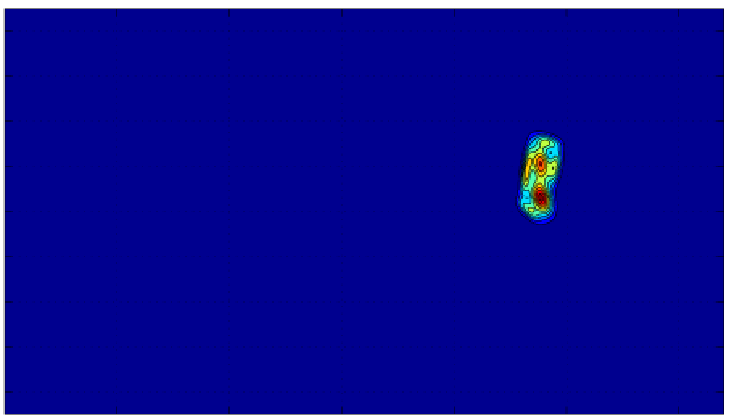}
\end{center}
\end{minipage}
\hspace{.2cm}
\begin{minipage}{4cm}
\begin{center}
\includegraphics[width=4cm,height=2.5cm]{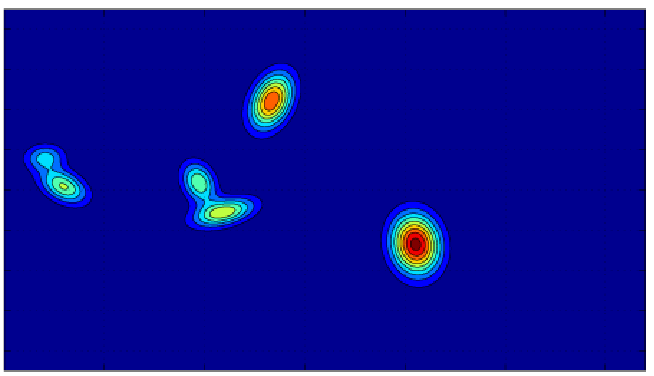}
\end{center}
\end{minipage}
\caption{Approximation of data at $t=0$s (left) and $t = 4$s (right).}
\label{fig:curve_fit}
\end{center}
\end{figure}

Then, the set $SFOV^c_i(R_{\V})$ is computed as
\begin{eqnarray}
SFOV^c_i(R_{\V}) = \bigcup_{j\in \V\setminus \{i\}} (\Q_{ij} \cap FOV_i \cap FOV_j).
\label{eqn:SFOV_check}
\end{eqnarray}
in polynomial time with respect to $n$.
Namely, checking $q_{wl}(\ewi) \in SFOV^c_i(R_{\V})$
for all $l\in \L_i$
provides  $\tilde \L_i(R_{\V})$.

The computation process including image processing, 
curve fitting and gradient computation
is successfully distributed to each sensor but the resulting FOVs need to be shared
among all sensors to compute $\tilde \L_i(R_{\V})$.
A way to implement the present scheme under 
limited communication
is to restrict the FOV of each sensor so that the FOV can overlap
with limited number of FOVs of the other sensors.
Such constraints on the FOVs are easily imposed by adding
an artificial potential to the objective function but 
 we leave the issue as a future work due to the page constraints.

\section{Simulation of Moving Objects Monitoring}

In this section, we demonstrate the utility of the
present approach through simulation using 4 cameras with $\lambda_i = 3.4$mm ${\forall i}$.
Here, we suppose that the view of the environment from $\Sigma_w$ with $\gamma = 10$m and focal length 
$3.4$mm is given as in 
Fig. \ref{fig:initial_image},
and that the mission space $\Q$ is equal to the FOV corresponding to the image.
Since the codes of simulating the image acquisition and processing are never used in experiments,
we simplify the process as follows, and demonstrate only the present coverage control scheme 
with the curve fitting process.
Before running the simulation, we compute the optical flows
for the images of Fig. \ref{fig:initial_image} as in Fig. \ref{fig:of_image},
and also fitting functions of the data as in Fig. \ref{fig:curve_fit}.
The resulting data is uploaded at
\url{http://www.fl.ctrl.titech.ac.jp/paper/2014/data.wmv}.
Then, we segment the image by the superlevel set
of the function using a threshold $10^{-3}$,
and assign a boolean variable $1$ to $y_{il}$ if 
$q_{wl}(\ewi)$ is inside of the set and assign $0$ otherwise.
The experimental system is now under construction, and the
experimental verification of the total process
will be conducted in a future work.
Note however that it is at least confirmed that 
the skipped image acquisition and processing 
can be implemented within several milliseconds in a real vision system.

\begin{figure}[t]
\begin{center}
\begin{minipage}{4cm}
\begin{center}
\includegraphics[width=4cm,height=2.5cm]{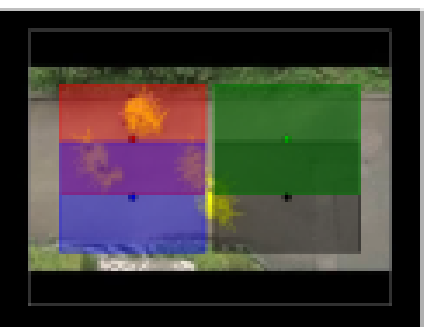}
\end{center}
\end{minipage}
\hspace{.2cm}
\begin{minipage}{4cm}
\begin{center}
\includegraphics[width=4cm,height=2.5cm]{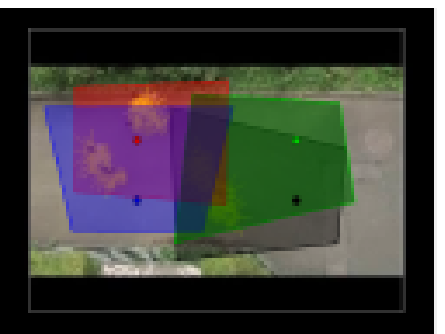}
\end{center}
\end{minipage}
\caption{Initial FOVs (left) and final FOVs (right).}
\label{fig:static}
\medskip

\includegraphics[width=4cm]{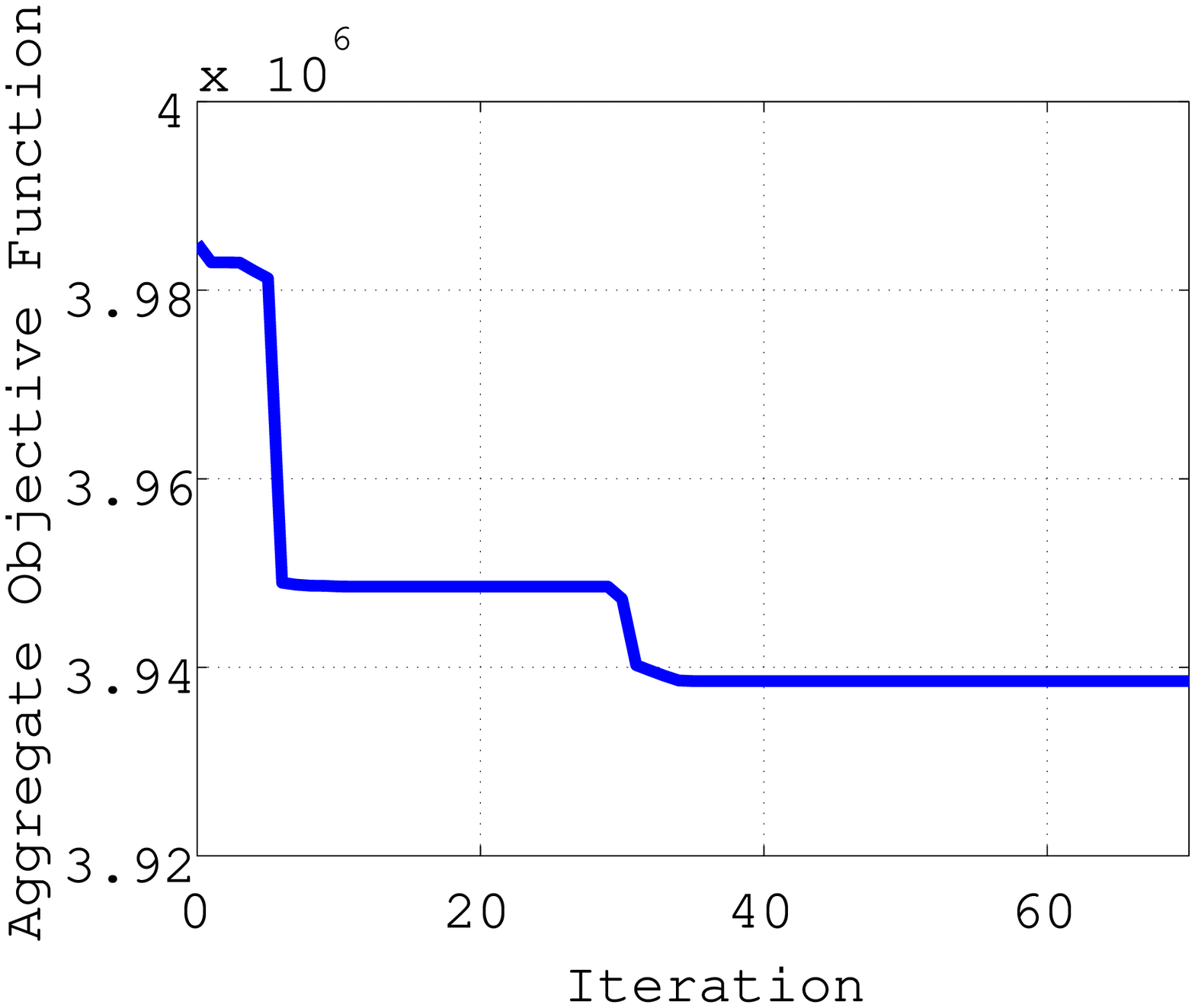}
\caption{Evolution of objective function $H$.}
\label{fig:obj}
\end{center}
\end{figure}

Let the position vectors of cameras be selected as $e_3^T p_{wi} = 6\mbox{m},\ {\forall i}$
and the length of each side of the image plane be $6.4$mm and $4.8$mm.
The other elements of $p_{wi}$ are set as illustrated by the mark $\bullet$ in Fig. \ref{fig:static}.
The parameters in $H$ is set as 
$\bar \psi = 1$, $\bar \phi = 1.05$, 
$w_{il} = \frac{\|p^v_{i1}\|+4\times 10^{-4}}{\|p_{il}\|+4\times 10^{-4}}$ 
and $W = {\rm diag}([0.01\ 0.01\ 1])$.
The curve fitting process
is run with $m = 3$ and the gradient is computed
by evaluating the objective function not at all points in $\L_i$
but at 121 points extracted from $\L_i$.
In order to confirm convergence of the orientations,
we first fix the image as in Fig. \ref{fig:static}
and run the present algorithm from the initial condition $\ewi = I_3\ {\forall i}$.
Then, the evolution of the function $H$ is illustrated in Fig. \ref{fig:obj},
where we compute the value using not the individually estimated density but
the data as in Fig. \ref{fig:curve_fit}.
We see from the figure that the function $H$ is decreasing
through the update process and eventually reaches a stationary point.
The final configuration is depicted in the right figure of Fig. \ref{fig:static}.
We next start to play the above movie and 
check adaptability to environmental changes,
where the orientations are assumed to be updated 
at each frame.
Then, the evolution of FOVs are shown in
\url{http://www.fl.ctrl.titech.ac.jp/paper/2014/sim.wmv}
whose snapshots at times $t = 0,1,2,4$ are depicted in Fig. \ref{fig:snaps}.
We see from the movie and figures that the cameras adjust their rotations
so as to capture moving humans.
The above results show the effectiveness of the present approach.

\section{Conclusions}

In this paper, we have investigated visual coverage control 
where the vision sensors
are assumed to be distributed over the 3-D space to 
monitor the 2-D environment
and to be able to control their orientations.
We first have formulated the problem as an optimization
problem on $SO(3)$.
Then, in order to solve the problem, we have presented
the entire process including not only the gradient computation
but also image processing and curve fitting, which are 
required to estimate the density function
from the acquired vision data.
Finally, we have demonstrated the effectiveness of the approach
through simulation of moving objects monitoring.

\appendix

\section{Proof of Theorem 1}
\label{app:1}

For notational simplicity, we describe $\ewi[k]$ by $R$ in the sequel.
Substituting (\ref{eqn:perf5}), (\ref{eqn:perf6}) and (\ref{eqn:mix_Gauss})
into (\ref{eqn:obj_single_fict}), the objective function to be minimized
is formulated as
\begin{eqnarray}
\!\!\!\!\! \bar H_i(M) \!\!&\!\!=\!\!&\!\! 
\tilde \delta_i \bar{\phi} \!\!
\sum_{l \in \tilde \L_i^c}\!\!w_{il}H_i^l+
\tilde \delta_i 
\!\!\sum_{l \in \tilde \L_i}\!\!w_{il}
\Big( \bar \psi H_i^l - \!\!
\sum_{j=1}^m \!\! \alpha_j H_i^{lj}\Big),
\label{eqn:obj_single_fict3}\\
\!\!\!\!\! H_i^l(M) \!\!&\!\!=\!\!&\!\!  \frac{ \|Mp_{il}\|_W^2}{\|{\bf e}_3^T M p_{il}\|^2},\
 H_i^{lj}(M) =
\frac{\|Mp_{il}\|_W^2}{\|{\bf e}_3^T M p_{il}\|^2}
E_i^{lj}(M),
\nonumber\\
\!\!\!\!\! E_i^{lj}(M) \!\!&\!\!=\!\!&\!\!
\exp\left\{-
\left\|
\frac{\lambda_i R^T Mp_{il}}
{{\bf e}_3^T R^T Mp_{il}}
- \mu_j \right\|^2_{\Sigma_j}
\right\}.
\nonumber
\end{eqnarray}

\begin{figure}[t]
\begin{center}
\begin{minipage}{4cm}
\begin{center}
\includegraphics[width=4cm,height=2.5cm]{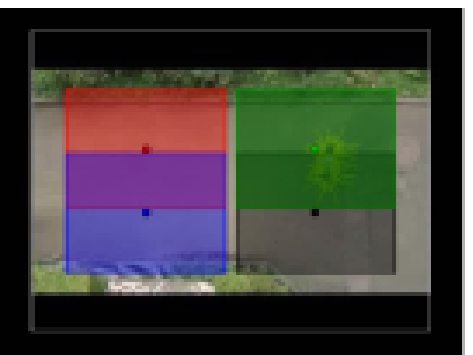}
\end{center}
\end{minipage}
\hspace{.2cm}
\begin{minipage}{4cm}
\begin{center}
\includegraphics[width=4cm,height=2.5cm]{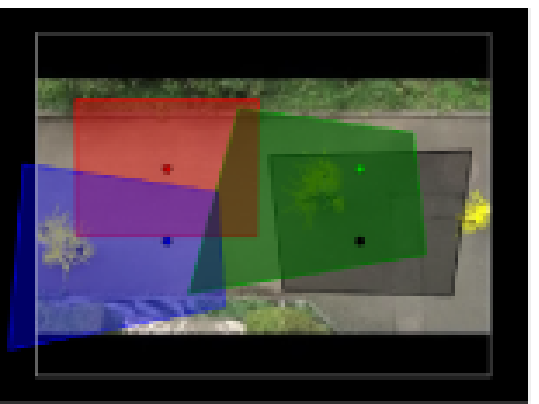}
\end{center}
\end{minipage}
\medskip

\begin{minipage}{4cm}
\begin{center}
\includegraphics[width=4cm,height=2.5cm]{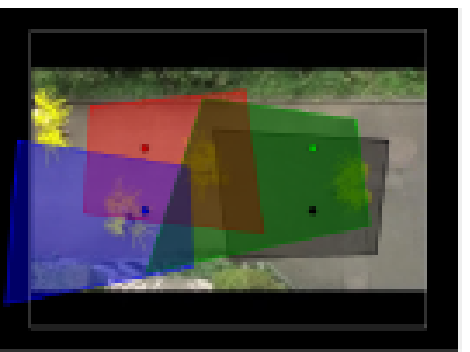}
\end{center}
\end{minipage}
\hspace{.2cm}
\begin{minipage}{4cm}
\begin{center}
\includegraphics[width=4cm,height=2.5cm]{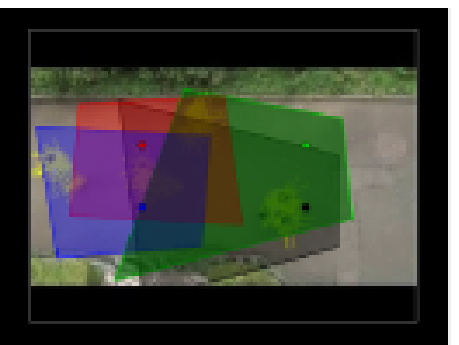}
\end{center}
\end{minipage}
\caption{Snapshots at $t=0$s (top-left), $t=1$s (top-right), $t=2$s (bottom-left) and $t = 4$s (bottom-right).}
\label{fig:snaps}
\end{center}
\end{figure}

From Lemma 1 and the fact that $SO(3)$
is a submanifold of $\R^{3\times 3}$,
we first compute the gradient $\grad_{R}^{\R^{3\times 3}} \bar H_i(R)$.
From Definition 1 and (\ref{eqn:derivative_lin_M}),
%
we need to compute the directional derivative
$D \bar H_i(R)[\Xi]$.
From linearity of the directional derivative operator $D$,
it is sufficient to derive 
$DH_i^l(R)[\Xi]$ and $DH_i^{lj}(R)[\Xi]$.

We first consider $DH_i^l(R)[\Xi]$.
By calculation, we have
\begin{eqnarray}
\hspace{-.8cm}&&\tilde H^{l}_i(R + \Xi t) - \tilde H^{l}_i(R)=
\frac{ \|(R+\Xi t)p_{il}\|_W^2}{\|{\bf e}_3^T (R+\Xi t) p_{il}\|^2}
- \frac{\bar \psi \|R p_{il}\|_W^2}{\|{\bf e}_3^T R p_{il}\|^2}
\nonumber\\
\hspace{-.8cm}&&\hspace{.2cm}
= \frac{\|{\bf e}_3^T R p_{il}\|^2\|(R+\Xi t)p_{il}\|_W^2}
{\|{\bf e}_3^T (R+\Xi t) p_{il}\|^2\|{\bf e}_3^T R p_{il}\|^2}
\nonumber\\
\hspace{-.8cm}&&\hspace{3.2cm}
- \frac{\|{\bf e}_3^T (R+\Xi t) p_{il}\|^2\|Rp_{il}\|_W^2}
{\|{\bf e}_3^T (R+\Xi t) p_{il}\|^2\|{\bf e}_3^T R p_{il}\|^2}
\nonumber\\
\hspace{-.8cm}&&\hspace{.2cm}
= 2t \frac{({\bf e}_3^T R p_{il})(p_{il}^TR^TW\Xi p_{il})
- \|Rp_{il}\|_W^2 ({\bf e}_3^T\Xi p_{il}) + o(t)}
{\|{\bf e}_3^T (R+\Xi t) p_{il}\|^2({\bf e}_3^T R p_{il})}
\nonumber
\end{eqnarray}
Hence, $D H_i^l(R)[\Xi] = \lim_{t\to 0}\frac{\tilde H^{l}_i(R + \Xi t) - \tilde H^{l}_i(R)}{t}$
is given by
\begin{eqnarray}
D H_i^l(M)[\Xi]
\!\!&\!\!=\!\!&\!\!
\tilde \eta_i^l(R)\Xi p_{il},
\label{eqn:modify2}\\
\eta_i^l(R) \!\!&\!\!=\!\!&\!\! \frac{2(
({\bf e}_3^T R p_{il})p_{il}^TR^TW
- \|Rp_{il}\|_W^2 {\bf e}_3^T
)}{({\bf e}_3^T R p_{il})^3}.
\nonumber
\end{eqnarray}

Let us next consider  
$D \bar H_i^{lj}(R)[\Xi]$. 
We first have the equations
\begin{eqnarray}
\hspace{-.8cm}&&\bar H^{lj}_i(R + \Xi t) =
\frac{\|(R+\Xi t) p_{il}\|_W^2E^{lj}_i(R + \Xi t)}{\|{\bf e}_3^T (R + \Xi t) p_{il}\|^2}
\nonumber\\
\hspace{-.8cm}&&\hspace{.3cm}= \frac{(\|Rp_{il}\|_W^2 + 2tp_{il}^TR^TW\Xi p_{il} + o(t^2))
E^{lj}_i(R + \Xi t)}{
\|{\bf e}_3^TR p_{il}\|^2 + 2t 
({\bf e}_3^TR p_{il})({\bf e}_3^T\Xi p_{il})
+ t^2 \|{\bf e}_3^T\Xi p_{il}\|^2.
}
\nonumber
\end{eqnarray}
Hence, we also have
\begin{eqnarray}
\hspace{-.8cm}&&\bar H^{lj}_i(R + \Xi t) - \bar H^{lj}_i(R)
\nonumber\\
\hspace{-.8cm}&&= \frac{(\|Rp_{il}\|_W^2 + 2tp_{il}^TR^TW\Xi p_{il} + o(t^2))E^{lj}_i(R + \Xi t)}{
\|{\bf e}_3^TR p_{il}\|^2 + 2t 
({\bf e}_3^TR p_{il})({\bf e}_3^T\Xi p_{il})
+ t^2 \|{\bf e}_3^T\Xi p_{il}\|^2}
\nonumber\\
\hspace{-.8cm}&&\hspace{3cm}- \frac{\|Rp_{il}\|_W^2E^{lj}_i(R)}{\|{\bf e}_3^T R p_{il}\|^2}
\nonumber\\
\hspace{-.8cm}&& =  
\frac{\left\{\|Rp_{il}\|_W^2 + 2tp_{il}^TR^TW\Xi p_{il}\right\}
\|{\bf e}_3^T R p_{il}\|^2
E^{lj}_i(R + \Xi t)}{ \|{\bf e}_3^TR p_{il}\|^4 + o(t)}
\nonumber\\
\hspace{-.8cm}&&
-\frac{\left\{\|{\bf e}_3^TR p_{il}\|^2 + 2t 
({\bf e}_3^TR p_{il})({\bf e}_3^T\Xi p_{il})\right\}
\|Rp_{il}\|_W^2E^{lj}_i(R)
}{\|{\bf e}_3^TR p_{il}\|^4 + o(t)}
\nonumber\\
\hspace{-.8cm}&&\hspace{3cm} + o(t^2)
\nonumber\\
\hspace{-.8cm}&& =  
\frac{\left\{\|Rp_{il}\|_W^2 + 2tp_{il}^TR^TW\Xi p_{il}\right\}
({\bf e}_3^T R p_{il})}{ ({\bf e}_3^TR p_{il})^3 + o(t)}E^{lj}_i(R + \Xi t)
\nonumber\\
\hspace{-.8cm}&&\hspace{1cm}
-\frac{\left\{({\bf e}_3^TR p_{il}) + 2t 
({\bf e}_3^T\Xi p_{il})\right\}
\|Rp_{il}\|_W^2
}{({\bf e}_3^TR p_{il})^3 + o(t)}E^{lj}_i(R)
\nonumber\\
\hspace{-.8cm}&&\hspace{3cm} + o(t^2).
\label{eqn:app1}
\end{eqnarray}
We also obtain
\begin{eqnarray}
E^{lj}_i(R + \Xi t) \!\!&\!\!=\!\!&\!\! 
\exp\left\{-
\left\|
\frac{\lambda_i 
(I_3 + tR^T\Xi)p_{il}}{{\bf e}_3^T (I_3 + tR^T\Xi )p_{il}}
- \mu_j \right\|^2_{\Sigma_j}\right\}
\nonumber\\
\!\!&\!\!=\!\!&\!\!  \exp\left\{-\left\|
 \frac{\lambda_i(b_{lj} + c_{lj}t)}{ \lambda_i +  {\bf e}_3^Ta_{l}t}
\right\|^2_{\Sigma_j}\right\},
\label{eqn:app3}
\end{eqnarray}
where $a_l =  R^T\Xi p_{il}$, 
$b_{lj} = p_{il} - \mu_j$ and $c_l = a_l - \frac{ {\bf e}_3^Ta_{l} }{\lambda_i}\mu_j$
are introduced for notational simplicity.
Using $e^a = \sum_{k=0}^{\infty}\frac{a^k}{k!}$,
we can decompose low and high order terms in $t$
as
\begin{eqnarray}
\hspace{-.8cm}&&E^{lj}_i(R + \Xi t) 
\nonumber\\
\hspace{-.8cm}&&= \exp\left\{-
\frac{\lambda_i^2\left(
\|b_{lj}\|^2_{\Sigma_j} + 2 tb_{lj}^T \Sigma_j c_{lj}
+ t^2 \|c_{lj}\|^2_{\Sigma_j}
\right)}{ (\lambda_i + {\bf e}_3^Ta_{l}t)^2}
\right\},
\nonumber\\
\hspace{-.8cm}&&= \sum_{k = 0}^{\infty}
\frac{1}{k!}
\left(
\frac{-\lambda_i^2}{ (\lambda_i + {\bf e}_3^Ta_{l}t)^2}
\right)^k\times
\nonumber\\
\hspace{-.8cm}&&\hspace{1.5cm}
\times
\left(
\|b_{lj}\|^2_{\Sigma_j} + 2 tb_{lj}^T \Sigma_j c_{lj}
+ t^2 \|c_{lj}\|^2_{\Sigma_j}
\right)^k
\nonumber\\
\hspace{-.8cm}&&=
1 + 
\sum_{k = 1}^{\infty}
\frac{1}{k!}
\left(
\frac{-\lambda_i^2}{ (\lambda_i + {\bf e}_3^Ta_{l}t)^2}
\right)^k\times
\nonumber\\
\hspace{-.8cm}&&\hspace{0.8cm}
\times
\left(
(\|b_{lj}\|^2_{\Sigma_j})^k + 2k t (\|b_{lj}\|^2_{\Sigma_j})^{k-1}(b_{lj}^T \Sigma_j c_{lj})
+ o(t^2)
\right)
\nonumber\\
\hspace{-.8cm}&&= \left(1 + 
\sum_{k = 1}^{\infty}
\frac{1}{k!}
\left(
\frac{-\lambda_i^2\|b_{lj}\|^2_{\Sigma_j}}{ (\lambda_i + {\bf e}_3^Ta_{l}t)^2}
\right)^k\right)+ o(t^2)
\nonumber\\
\hspace{-.8cm}&&\hspace{.4cm}-
\frac{2t\lambda_i^2(b_{lj}^T \Sigma_j c_{lj})}{ (\lambda_i + {\bf e}_3^Ta_{l}t)^2}
\left(
\sum_{k = 1}^{\infty}
\frac{1}{(k-1)!}
\left(
\frac{-\lambda_i^2\|b_{lj}\|^2_{\Sigma_j}}{ (\lambda_i + {\bf e}_3^Ta_{l}t)^2}
\right)^{k-1}\right)
\nonumber\\
\hspace{-.8cm}&&=
\left(
1 - \frac{2\lambda_i^2(b_{lj}^T \Sigma_j c_{lj})t}{ (\lambda_i + {\bf e}_3^Ta_{l}t)^2}
\right)
\left(\sum_{k = 0}^{\infty}
\frac{1}{k!}
\left(
\frac{-\lambda_i^2\|b_{lj}\|^2_{\Sigma_j}}{ (\lambda_i + {\bf e}_3^Ta_{l}t)^2}
\right)^k\right) 
\nonumber\\
\hspace{-.8cm}&&\hspace{4.4cm}
+ o(t^2)
\label{eqn:app4}
\end{eqnarray}
(\ref{eqn:app4}) is also simplified as
\begin{eqnarray}
\hspace{-.8cm}&&E^{lj}_i(R + \Xi t) =
h_i^{lj}(t)+o(t^2)
\label{eqn:app4}\\
\hspace{-.8cm}&&h_i^{lj}(t) :=
\left(
1 - \frac{2\lambda_i^2(b_{lj}^T \Sigma_j c_{lj})t}{ (\lambda_i + {\bf e}_3^Ta_{l}t)^2}
\right)
\exp\left\{
\frac{-\lambda_i^2\|b_{lj}\|^2_{\Sigma_j}}{ (\lambda_i + {\bf e}_3^Ta_{l}t)^2}
\right\},
\nonumber
\end{eqnarray}
where
\begin{eqnarray}
h_i^{lj}(0) = 
\exp\left\{-\|b_{lj}\|^2_{\Sigma_j}
\right\} 
= E_i^{lj}(R).
\label{eqn:app6}
\end{eqnarray}
Substituting (\ref{eqn:app4}) and (\ref{eqn:app6}) 
into (\ref{eqn:app1}) yields
\begin{eqnarray}
\hspace{-.8cm}&& H^{lj}_i(R + \Xi t) -  H^{lj}_i(R)
=  
\frac{\|Rp_{il}\|_W^2({\bf e}_3^T R p_{il})}
{({\bf e}_3^TR p_{il})^3 + o(t)}\times
\nonumber\\
\hspace{-.8cm}&& \hspace{.5cm} 
\times (h_i^{lj}(t) - h_i^{lj}(0))
+ 2t\frac{
(p_{il}^TR^TW\Xi p_{il})
({\bf e}_3^T R p_{il})h_i^{lj}(t)
}{({\bf e}_3^TR p_{il})^3 + o(t)}
\nonumber\\
\hspace{-.8cm}&& \hspace{.5cm} 
- 2t\frac{
({\bf e}_3^T\Xi p_{il})
\|Rp_{il}\|_W^2 h_i^{lj}(0)
}{({\bf e}_3^TR p_{il})^3 + o(t)}
+ o(t^2)
\label{eqn:app7}
\end{eqnarray}

Let us now compute $D  H^{lj}_i(R)[\Xi]$.
Substituting (\ref{eqn:app7}) into the definition of the directional derivative (\ref{eqn:derivative_lin_M}),
i.e.
\begin{eqnarray}
D \bar H^{lj}_i(R)[\Xi]
\!\!&\!\!=\!\!&\!\!\lim_{t \to 0}\frac{\bar H^{lj}_i(R + \Xi t) - \bar H^{lj}_i(R)}{t},
\label{eqn:app8}
\end{eqnarray}
we have
\begin{eqnarray}
\hspace{-1.2cm}&&D  H^{lj}_i(R)[\Xi]
=
\frac{\|Rp_{il}\|_W^2({\bf e}_3^T R p_{il})}
{({\bf e}_3^TR p_{il})^3}
\left(\frac{d h_i^{lj}}{dt}\right)(0)
\nonumber\\
\hspace{-1.2cm}&& \hspace{.2cm} + 2h_i^{lj}(0)
\frac{
({\bf e}_3^TR p_{il})p_{il}^TR^TW\Xi p_{il}
-\|Rp_{il}\|_W^2{\bf e}_3^T\Xi p_{il}
}{({\bf e}_3^TR p_{il})^3}
\label{eqn:app9}
\end{eqnarray}
By calculation, the derivative $\frac{d h_i^{lj}}{dt}$ is given by
\begin{eqnarray}
\hspace{-.8cm}&&\frac{d h_i^{lj}}{dt} = 
\exp\left\{
\frac{-\lambda_i^2\|b_{lj}\|^2_{\Sigma_j}}{ (\lambda_i + {\bf e}_3^Ta_{l}t)^2}
\right\}
\frac{1}{(\lambda_i + {\bf e}_3^Ta_{l}t)^3}\times
\nonumber\\
\hspace{-.8cm}&&\times
\Big( 2 \lambda_i^2\|b_{lj}\|^2_{\Sigma_j} {\bf e}_3^Ta_{l} 
- \frac{4 \lambda_i^4(b_{lj}^T \Sigma_j c_{lj})\|b_{lj}\|^2_{\Sigma_j} {\bf e}_3^Ta_{l}  }
{(\lambda_i + {\bf e}_3^Ta_{l}t)^2}t
\nonumber\\
\hspace{-.8cm}&&\hspace{.5cm}
- 2\lambda_i^3(b_{lj}^T \Sigma_j c_{lj}) + 2 \lambda_i^2(b_{lj}^T \Sigma_j c_{lj}){\bf e}_3^Ta_{l} t
\Big)
\end{eqnarray}
and hence
\begin{eqnarray}
\hspace{-0.5cm}
\frac{d h_i^{lj}}{dt}(0) \!\!&\!\!=\!\!&\!\! 
\frac{2e^{-\|b_{lj}\|^2_{\Sigma_j}}}{\lambda_i}
\Big( \|b_{lj}\|_{\Sigma_j}^2 {\bf e}_3^Ta_{l} 
- \lambda_i(b_{lj}^T \Sigma_j c_{lj})
\Big).
\label{eqn:app10}
\end{eqnarray}
Substituting (\ref{eqn:app10}) and definitions of $a_l$ and $c_l$
into (\ref{eqn:app9})
yields
\begin{eqnarray}
\hspace{-.8cm}&&D H^{lj}_i(R)[\Xi]
= \bar \eta_{i}^{lj}(R) \Xi p_{il}
\label{eqn:app13}\\
\hspace{-.8cm}&&\eta_{i}^{lj}(R) =
\frac{2e^{-\|b_{lj}\|^2_{\Sigma_j}}}
{\lambda_i({\bf e}_3^TR p_{il})^3}
\Big(({\bf e}_3^TR p_{il})\xi^{lj}_i(R)
-\lambda_i\|Rp_{il}\|_W^2{\bf e}_3^T
\Big)
\nonumber\\
\hspace{-.8cm}&&\xi^{lj}_i(R) =
 \|Rp_{il}\|_W^2b_{lj}^T \Sigma_j (p_{il} {\bf e}_3^T - \lambda_i I_3)R^T
+ \lambda_i p_{il}^TR^TW.
\nonumber
\end{eqnarray}
Note that $b_{lj} = p_{il} - \mu_j$ is constant and 
$\eta_{i}^{lj}(R)$ is independent of the matrix $\Xi$.

From (\ref{eqn:obj_single_fict3}), (\ref{eqn:modify2}) and (\ref{eqn:app13}),
we obtain
\begin{eqnarray}
\hspace{-.8cm}&&D\bar H_i(R)[\Xi] = \tilde \delta_i 
\eta_i(R) \Xi p_{il}
=\tr \Big(\Xi^T \Big(\tilde \delta_i
\eta_i^T(R) p_{il}^T \Big)\Big),
\nonumber\\
\hspace{-.8cm}&&\eta_i(R) =
\sum_{l \in \tilde \L_i^c(R)}w_{il}
\bar{\phi}  \eta_i^l+\!\!
\sum_{l \in \tilde \L_i(R)}\!\! w_{il}
\Big(
\bar{\psi} \eta_i^l
- \sum_{j=1}^m \alpha_j \eta_{i}^{lj}
\Big).
\nonumber
\end{eqnarray}
From Definition \ref{def:grad_M},
we have
$\grad_{R}^{\R^{3\times 3}}\bar H_i = \tilde \delta_i \eta_i^T(R) p_{il}^T$.
Combining it with Lemma 1 and
(\ref{eqn:proj_so(3)}) completes the proof.


\begin{thebibliography}{99}




\bibitem{amit2}
B. Song, C. Ding, A. Kamal, J. A. Farrell and A. Roy-Chowdhury,
``Distributed camera networks: integrated sensing and analysis for wide
area scene understanding,'' {\it IEEE Signal Processing Magazine}, Vol. 28,
No. 3, pp. 20--31, 2011.

\bibitem{TAC13}
T. Hatanaka and M. Fujita,
``Cooperative Estimation of Averaged 3D Moving Target Object Poses via Networked Visual Motion Observers,'' 
{\it IEEE Trans. Automatic Control}, Vol. 58, No. 3, pp. 623--638, 2013.



\bibitem{TV_SP11}
R. Tron and R. Vidal,
``Challenges faced in deployment of camera sensor networks,''
{\it IEEE Signal Processing Magazine}, Vol. 28, No. 3, pp. 32--45, 2011.


\bibitem{EYE}
B. M. Schwager, B. J. Julian, M. Angermann and D. Rus,
``Eyes in the Sky: Decentralized
Control for the Deployment of Robotic Camera Networks,''
{\it Proc. of the IEEE}, Vol. 99, No. 9, pp. 1541--1641, 2011.


\bibitem{cortes}
K. Laventall and J. Cortes,
``Coverage control by robotic networks with limited-range anisotropic sensory,''
{\it Proc. of 2008 American Control Conf.}, pp. 2666--2671, 2008.

\bibitem{GTF_CDC08}
A. Gusrialdi, T. Hatanaka and M. Fujita,
``Coverage control for mobile networks with limited-range anisotropic sensors,''
{\it Proc. of 47th IEEE Conf. on Decision and Control}, pp. 4263--4268, 2008.

\bibitem{PhD}
A. Ganguli, {\it Motion coordination for mobile robotic networks with visibility sensors}, 
{\it PhD thesis, Electrical and Computer Engineering Department, University of Illinois at Urbana-Champaign}, 2007.


\bibitem{ZM_SIAM13}
M. Zhu and S. Martinez,
``Distributed coverage games for energy-aware mobile sensor networks,''
{\it SIAM J. Control and Optimization}, Vol. 51, No. 1, pp. 1--27, 2013.


\bibitem{DSMFC_SP12}
C. Ding, B. Song, A. Morye, J. A. Farrell and A. K. Roy-Chowdhury,
``Collaborative sensing in a distributed PT camera network,''
{\it IEEE Trans. Image Processing}, Vol. 21, No. 7, pp. 3282--3295, 2012.



\bibitem{HWF_CDC13}
T. Hatanaka, Y. Wasa and M. Fujita
Game Theoretic Cooperative Control of PTZ Visual Sensor Networks for Environmental Change Monitoring
{\it Proc. of 52nd IEEE Conf. on Decision and Control}, to appear, 2013





\bibitem{CL_EJC05}
C. G. Cassandras and W. Li, 
``Sensor networks and cooperative control,''
{\it Euro. J. Control}, Vol. 11, No. 4--5, pp. 436--463, 2005.


\bibitem{BCM_BK09}
S. Martinez, J. Cortes, and F. Bullo,
``Motion coordination with distributed information,'' 
{\it IEEE Control Systems Magazine}, Vol. 27, No. 4, pp. 75--88, 2007.

\bibitem{BCM_ES05}
J. Cortes, S. Martinez, and F. Bullo, 
``Spatially-distributed coverage optimization and control with limited-range interactions,'' 
ESAIM: Control, Optimisation \& Calculus of Variations, Vol. 11, pp. 691--719, 2005.

\bibitem{IJRR}
M. Schwager, D. Rus, and J. J. Slotine,
``Decentralized, adaptive coverage
control for networked robots,'' 
{\it Int. J. Robotic Research}, Vol. 28, No. 3, pp. 357--375,
2009.




\bibitem{survey}
A. Mavrinac and X. Chen,
``Modeling Coverage in Camera Networks: A Survey
International J. Computer Vision,''
Vol. 101, No. 1, pp. 205--226, 2013.





\bibitem{AMS_BK}
P. A. Absil, R. Mahony and R. Sepulchre, {\it Optimization Algorithms on
Matrix Manifolds}, Princeton University Press, 2008.











\bibitem{MATLAB1}
Mathworks, {\it Computer Vision Toolbox User's Guide}, 2013.

\bibitem{MATLAB2}
Mathworks, {\it Curve Fitting Toolbox User's Guide}, 2013.

\bibitem{MLS_BK}
R. Murray, Z. Li, and S. S. Sastry, {\it A Mathematical Introduction to
Robotic Manipulation}, CRC Press, 1994.
\end{thebibliography}
\end{document}